\documentclass[preprint,11pt]{elsarticle}
\makeatletter
\def\ps@pprintTitle{%
 \let\@oddhead\@empty
 \let\@evenhead\@empty
 \def\@oddfoot{\centerline{\thepage}}%
 \let\@evenfoot\@oddfoot}
\makeatother

\usepackage{tikz,calc}
\usepackage{subcaption}
\usepackage{mathrsfs}
\usepackage{color}
\usepackage{array}
\usepackage{amssymb}                           
\usepackage{amsmath}                           
\usepackage[linesnumbered,vlined,ruled,algo2e]{algorithm2e}

\SetCommentSty{mycommfont}

\newcommand{\mscr}[1]{\mathscr{#1}}
\newcommand{\mrm}[1]{\mathrm{#1}}
\newcommand{\mtt}[1]{\mathtt{#1}}
\newcommand{\msf}[1]{\mathsf{#1}}
\renewcommand{\mit}[1]{\mathit{#1}}
\newcommand{\lrangle}[1]{\langle #1 \rangle}

\newtheorem{theorem}{Theorem}
\newtheorem{lemma}{Lemma}
\newtheorem{corollary}{Corollary}
\newtheorem{proposition}{Proposition}
\newtheorem{observation}{Observation}
\newdefinition{example}{Example}
\newdefinition{defn}{Definition}
\newdefinition{condition}{Condition}
\newproof{proof}{Proof}

\newcommand{\rev}[1]{#1^\mathrm{R}}
\newcommand{\prev}[1]{#1^{\triangleleft}}
\newcommand{\plink}{\sqsubset}
\newcommand{\rad}{\rho}

\newcommand{\Stack}{\mathit{Stack}}
\newcommand{\Pals}{\mit{Pals}}

\newcommand{\ot}{:=}
\newcommand{\PCE}{\msf{Stabilize}}
\newcommand{\slowExt}{\msf{SlowExtend}}
\newcommand{\fastExt}{\msf{FastExtend}}
\newcommand{\unstable}{\mit{unstable}}
\newcommand{\ZRed}{\msf{ZReduce}}
\newcommand{\ZDet}{\msf{ZDetect}}
\newcommand{\ZDetF}{\msf{ZDetectChain}}
\newcommand{\ACE}{\textbf{\AE}}
\newcommand{\Append}{\msf{append}}
\newcommand{\ZR}{\mathrel{\rightarrow}}
\newcommand{\Rad}[2]{\rho_{#1}(#2)}
\newcommand{\fst}[1]{\textcolor{blue}{#1}}
\newcommand{\snd}[1]{\textcolor{red}{#1}}

\SetArgSty{textrm}

\newcommand{\palin}[4][]{
	\draw[#1] (#2-#3,#4) -- (#2,#4) -- (#2+#3,#4);
	\draw[#1] (#2-#3,#4-0.2) -- (#2-#3,#4+0.2);
	\draw[#1] (#2,#4-0.2) -- (#2,#4+0.2);
	\draw[#1] (#2+#3,#4-0.2) -- (#2+#3,#4+0.2);
}
\newcommand{\zshape}[4][]{
	\draw[#1] (2*#2-#3,#4) -- (#2,#4) -- (#3,#4) -- (2*#3-#2,#4);
	\draw[#1] (2*#2-#3,#4-0.2) -- (2*#2-#3,#4+0.2);
	\draw[#1] (#2,#4-0.2) -- (#2,#4+0.2);
	\draw[#1] (#3,#4-0.2) -- (#3,#4+0.2);
	\draw[#1] (2*#3-#2,#4-0.2) -- (2*#3-#2,#4+0.2);
}

\begin{document}
	
\begin{frontmatter}
		
	\title{Linear-Time Online Algorithm for  \\ Inferring the Shortest Path from a Walk}
	\author[Tohoku]{Shintaro~Narisada\footnote{He is currently working in KDDI Corporation, Tokyo, Japan.}}
	\author[Tohoku]{Diptarama~Hendrian}
	\author[Tohoku]{Ryo~Yoshinaka}
	\author[Tohoku]{Ayumi~Shinohara}
	\address[Tohoku]{Graduate School of Information Sciences, Tohoku University, Sendai, Japan}	
		
\begin{abstract}
We consider the problem of inferring an edge-labeled graph from the sequence of edge labels seen in a walk on that graph.
It has been known that this problem is solvable in $\mrm{O}(n \log n)$ time when the targets are path or cycle graphs. 
This paper presents an online algorithm for the problem of this restricted case that runs in $\mrm{O}(n)$ time, based on Manacher's algorithm for computing all the maximal palindromes in a string.
\end{abstract}


\end{frontmatter}
\section{Introduction}
Aslam and Rivest~\cite{Aslam1990359} proposed the problem of \emph{minimum graph inference from a walk}.
Let us consider an edge-labeled undirected (multi)graph $G$.
A \emph{walk} of $G$ is a sequence of edges $e_1,\dots,e_n$ such that each $e_i$ connects $v_{i-1}$ and $v_i$ for some (not necessarily pairwise distinct) vertices $v_0,v_1,\dots,v_n$.
The \emph{output} of the walk is the sequence of the labels of those edges. 
For a string $w$, \emph{minimum graph inference from a walk} is the problem to compute a graph $G$ with the smallest number of \emph{vertices} such that $w$ is the output of a walk of $G$.
We give an example in Figure~\ref{figure:z}.
\begin{figure}[t]
	\centering
	\newcommand{\paline}[4]{
	\draw (#1,#4) -- (#2,#4) -- (#3,#4);
	\draw (#1,#4-0.2) -- (#1,#4+0.2);
	\draw (#2,#4-0.2) -- (#2,#4+0.2);
	\draw (#3,#4-0.2) -- (#3,#4+0.2);
	}
	\begin{tikzpicture}[scale=0.8,every node/.style={circle},inner sep=0pt,minimum size=2.5mm]\small
		\node[draw,fill=black] (1) at (0,1.8) {}; 
		\node[draw,fill=black] (2) at (1.5,1.8) {}; 
		\node[draw,fill=black] (3) at (3,1.8) {}; 
		\node[draw,fill=black] (4) at (4.5,1.8) {}; 
		\node[draw,fill=black] (5) at (6,1.8) {}; 
		\draw[-,thick] (1) to node[above] {\Large$\mtt{a}$} (2);
		\draw[-,thick] (2) to node[above] {\Large$\mtt{b}$} (3);
		\draw[-,thick] (3) to node[above] {\Large$\mtt{c}$} (4);
		\draw[-,thick] (4) to node[above] {\Large$\mtt{a}$} (5);
		\draw[->,densely dashed,thick,red] (0,1.2) -- (6,1.2) -- (1.5,0.9) -- (3,0.7) -- (0,0.4) -- (4.5,0.4) -- (0,0) -- (6,0);
	\end{tikzpicture}
	\caption{Minimum path graph that has ${\tt abcaacbbbaabccbbca}$ as a walk output\label{figure:z}}
\end{figure}
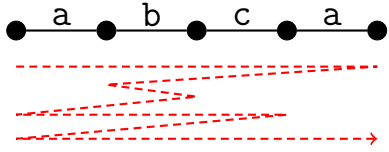
With no assumption on graphs to infer, trivially the graph with a single vertex with self-loops labeled with all output letters is always minimum.
The problem has been studied for different graph classes in the literature.

Aslam and Rivest~\cite{Aslam1990359} proposed polynomial time algorithms for the minimum graph inference problem for path graphs and cycle graphs,
which include the variant of minimum path graph inference where a walk must start from an end of a path graph and end in the other end (Table~\ref{table:inferres}).
Raghavan~\cite{RAGHAVAN1994108} studied the problem further and showed that both minimum path and cycle graph inference from walk are reduced to path graph inference from an end-to-end walk in $\mrm{O}(n)$ time.
Moreover, he presented an $\mrm{O}(n \log n)$ time algorithm for inferring minimum path/cycle graph from a walk, while showing inferring minimum graph with bounded degree $k$ is NP-hard for any $k \ge 3$.
Maruyama and Miyano~\cite{MARUYAMA1995257} strengthened Raghavan's result so that inferring minimum tree with bounded degree $k$ is still NP-hard for any $k \ge 3$.
On the other hand, Maruyama and Miyano~\cite{MARUYAMA1996289} showed that it is solvable in linear time when trees have no degree bound.
They also studied a variant of the problem where the input consists of multiple path labels rather than a single walk label, which was shown to be NP-hard.
Akutsu and Fukagawa~\cite{Akutsu2005371} considered another variant, where the input is the numbers of occurrences of vertex-labeled paths.
They showed a polynomial time algorithm with respect to the size of output graph, when
the graphs are trees of unbounded degree and the lengths of given paths are fixed.  
They also proved that the problem is strongly NP-hard even when the graphs are planar of unbounded degree.
\begin{table}[b]
	\caption{Time complexity of minimum graph inference bounded degree 2 from a walk \label{table:inferres}}
	\vspace{2mm}
	\centering
	\begin{tabular}{|r|r|r|r|} \hline
		& \multicolumn{3}{ c| }{Connected graph bounded degree $2$}  \\ 
		\cline{2-4}
		\multicolumn{1}{ |c| }{Algorithms}   & \multicolumn{2}{ c| }{path} &  \multicolumn{1}{ c| }{cycle} \\ 
		\cline{2-3}
		& end-to-end walk & general walk &  \\
		\hline
		\multicolumn{1}{ |c| }{Aslam \& Rivest~\cite{Aslam1990359}} & $\mrm{O}(n^3)$ & $\mrm{O}(n^3)$ & $\mrm{O}(n^5)$ \\ 
		\hline
		\multicolumn{1}{ |c| }{Raghavan~\cite{RAGHAVAN1994108}}& $\mrm{O}(n \log n)$ & $\mrm{O}(n \log n)$ & $\mrm{O}(n \log n)$ \\ 
		\hline
		\multicolumn{1}{ |c| }{Proposed}& $\mrm{O}(n)$ & $\mrm{O}(n)$ & $\mrm{O}(n)$ \\ 
		\hline
	\end{tabular}
\end{table}

This paper focuses on the problem on graphs of bounded degree 2, i.e., path and cycle graphs.
We propose a linear-time online algorithm that infers the minimum path graph from an end-to-end walk.
Thanks to Raghavan's result~\cite{RAGHAVAN1994108}, this entails that one can infer the minimum path/cycle graph in linear time from a walk, which is not necessarily end-to-end.
Aslam and Rivest~\cite{Aslam1990359} showed that the minimum path graphs that have end-to-end walks $x y \rev{y} y z$ and $xyz$ coincide, where $x,y,z$ are label strings and $\rev{y}$ is the reverse of $y$.
Let us call a nonempty string of the form $y \rev{y} y$ a \emph{Z-shape}.
Their result implies that to obtain the minimum path graph of a label string, one can repeatedly contract an arbitrary occurrence of a Z-shape $y \rev{y} y$ to $y$ until the sequence contains no such substring.
Then the finally obtained string is just the sequence of labels of the edges of the minimum path graph.
Raghavan~\cite{RAGHAVAN1994108}  achieved an $\mrm{O}(n \log n)$ time algorithm by introducing a sophisticated order of rewriting, which always contract the smallest Z-shapes in the sequence.
We follow their approach of repetitive contraction of Z-shapes but with a different order.
The order we take might appear more naive;
We read letters of the input string one by one and always contract the firstly found Z-shape.
This approach makes our algorithm online.
Apparently finding Z-shapes is closely related to finding palindromes.
Manacher~\cite{Manacher75} presented a linear-time ``online'' algorithm that finds all the maximal palindromes in a string.
To realize linear-time Z-shape elimination, we modify Manacher's algorithm for Z-shape detection and elimination, though it is not a straightforward adjustment.
Our experimental results show that our algorithm is faster than Raghavan's in practice, too.

A preliminary version of this paper appears in~\cite{Narisada2018}.
\section{Preliminaries}
For a tuple $\vec{e}=(e_1,\dots,e_m)$ of elements, we represent $(e_0,e_1,\dots,e_m)$ by $e_0;\vec{e}$ or $(e_0;\vec{e})$.
The interval between two integers $i$ and $j$ is denoted by $[i:j] = \{\, k \in \mathbb{Z} \mid i \le k \le j\,\}$.

Let $\Sigma$ be an alphabet.
A sequence of elements of $\Sigma$ is called a \emph{string} and the set of strings is denoted by $\Sigma^*$.
The empty string is denoted by $\varepsilon$ and the set of nonempty strings is $\Sigma^+ = \Sigma^* \setminus \{\varepsilon\}$.
For a string $w=xyz$, $x$, $y$, and $z$ are called a \emph{prefix}, a \emph{substring}, and a \emph{suffix} of $w$, respectively.
A prefix $x$ of $w$ is said to be \emph{proper} if $x \neq w$.
The length of $w$ is denoted by $|w|$.
The $i$-th letter of $w$ is denoted by $w[i]$ for $1 \leq i \leq |w|$. 
For $1 \leq i \leq j \leq |w|$, ${w}[i:j]$ represents the string $w[i] \dots w[j]$.
If $j > i$, $w[i:j]$ means the empty string. 
The longest proper prefix $w[1:|w|-1]$ of $w$ is abbreviated as $\prev{w}$ if $w \neq \varepsilon$.
The reversed string of $w$ is denoted by $\rev{w} = w[|w|] \cdots w[1]$.
The string repeating $w$ $k$ times is $w^k$.

A string $y$ is called an \emph{even palindrome} if $y = x\rev{x}$ for a string $x \in \Sigma^*$.
The \emph{radius} of $y$ is $r=|x|$.
Throughout this paper by a palindrome, we exclusively mean an even palindrome,  
 because we consider only even palindromes in this paper.
When $y$ occurs as a substring $w[i:j]$ of a string $w$, the position $c = (i+j-1)/2$ is called the \emph{center} (of the occurrence) of $y$.\footnote{%
It may be more reasonable to define the center to be $(i+j)/2$, but we have chosen to stick to integers.}
Especially, $y$ is said to be the \emph{maximal palindrome} (centered) at $c$
iff either  $i = 1$, $j = |w|$, or $w[i-1] \neq w[j+1]$.
The radius of  the maximal palindrome centered at $c$ in $w$ is denoted by $\Rad{w}{c}$.
The intervals $[c-\rad_w(c)+1:c]$ and $[c+1 : c+\rad_w(c)]$ are called the \emph{left} and \emph{right arms} of the maximal palindrome centered at $c$, respectively.  

A string $z$ is called a \emph{Z-shape}
if $z = x\rev{x}x$ for a non-empty string $x \in \Sigma^+$.
The \emph{tail} of $z$ is the suffix $\rev{x} x$. 
When $z$ occurs as a substring $z = w[i:i+3|x|-1]$ of a string $w$, the positions $p_1 = i+|x|-1$ and $p_2 = i+2|x|-1$ are called the \emph{left} and \emph{right pivots} (of the occurrence) of $z$.
The occurrence of the Z-shape is represented by a pair $\lrangle{p_1,p_2}$.
Note that the left and right pivots are the centers of the constituent palindromes $x\rev{x}$ and $\rev{x}x$, respectively.
Obviously, a pair $\lrangle{p_1,p_2}$ of positions in $w$ is a Z-shape occurrence if and only if $\rad_w(p_1),\rad_w(p_2) \ge p_2-p_1 > 0$.
We note that the empty string $\varepsilon$ is a palindrome but not a Z-shape by definition.
\begin{example}
Let us consider the string $w=\mtt{ababccbaabcc}$ illustrated below.
\[	\begin{tikzpicture}[scale=0.7]
	\foreach \x in {1,...,12}
		\draw(\x*0.5+0.35, 0.2) node{\scriptsize$\x$};
	\draw(0,-0.2) node{$w =$};
	\foreach \x/\s in {1/a,2/b,3/a,4/b,5/c,6/c,7/b,8/a,9/a,10/b,11/c,12/c}
		\draw(\x*0.5, -0.2) node[anchor=west]{\normalfont$\mathstrut\mathtt{\s}$};
	\palin{3.1}{1.5}{-0.8}
	\palin{4.6}{2}{-1.3}
	\zshape[thick]{3.1}{4.6}{-1.8}
\end{tikzpicture}\]
	We can find among others two maximal palindromes $w[3:8]=\mtt{abccba}$ and $w[5:12]=\mtt{ccbaabcc}$ whose centers are $5$ and $8$ and radii are  $\rad_w(5)=3$ and $\rad_w(8)=4$, respectively.
	Those two palindromes form a Z-shape $w[3:11]=\mtt{abccbaabc}$, whose occurrence is denoted by $\lrangle{5,8}$.
\end{example}

\subsection*{Minimum graph inference from a walk}
Let us define a binary relation $\to$ over nonempty strings by $x y \rev{y} y z \ZR xyz$ for $x,z \in \Sigma^*$ and $y \in \Sigma^+$,
saying that $x y \rev{y} y z$ \emph{reduces} to $xyz$. 
We call a string $w$ \emph{reducible} if it admits a string $w'$ such that $w \ZR w'$. Otherwise it is \emph{irreducible}.
In general, there can be different strings to which $w$ can reduce. 
Aslam and Rivest~\cite{Aslam1990359} proved that every string $w$ admits a unique irreducible string $w'$ such that $w \ZR^* w'$, where $\ZR^*$ is the reflexive and transitive closure of $\ZR$, which is obtained by repeatedly reducing $w$ by an arbitrary order.
Let us call the string $w'$ the \emph{Z-normal form} of $w$ and denote it by $\hat{w}$.
Their result can be written as follows.
\begin{theorem}[\cite{Aslam1990359}] \label{theo:end}
	The sequence of the labels of the edges of the minimum path graph with output $T$ of an end-to-end walk is its Z-normal form $\widehat{T}$.  
\end{theorem}
Therefore, to infer the minimum path graph from an end-to-end walk is to calculate its Z-normal form.
\begin{example}\label{ex:toy}
	The Z-normal form of $T = \mathtt{cbaaaabccbaabba}$
	is $\widehat{T}=\mathtt{cba}$, which is obtained by
	$\mathtt{cb \underline{aaa}abccbaabba} \ZR 
	\mathtt{\underline{cbaabccba}abba} \ZR 
	\mathtt{c\underline{baabba}} \ZR \mathtt{cba}$.
	Here, underlines show Z-shapes to contract.
	Another way to obtain $\widehat{T}$ is
	$\mathtt{cb aaaabcc\underline{baabba}} \linebreak[0]  \ZR 
	\mathtt{cb a\underline{aaa}bcc ba} \ZR 
	\mathtt{\underline{cbaabccba}} \ZR \mathtt{cba}$.
\end{example}


\section{Z-shape reduction algorithm}

We call a string $w$ \emph{pp-irreducible} (proper prefix is irreducible) if its longest proper prefix $\prev{w}=w[1:|w|-1]$ is irreducible.
A string $w$ is said to be \emph{ss-reducible} (solely suffix is reducible) if it is reducible and pp-irreducible.
Clearly a Z-shape occurs in an ss-reducible string as a suffix. By deleting its tail, we obtain an irreducible string.
A pp-irreducible string is either irreducible or ss-reducible.
Strings our online algorithm handles are all pp-irreducible.

Starting with $w=v_0=\varepsilon$, our algorithm repeats the following procedure for $i=1,2,\dots$.
We extend $w = v_{i-1}$ by reading letters from the input string $T$ one by one until it becomes an ss-reducible string $w = u_i$.
Then we reduce $u_i$ to $v_i = \widehat{u}_i$ by deleting the tail of the Z-shape and resume reading letters of $T$.
By repeatedly applying the procedure, we finally obtain the Z-normal form $w = \widehat{T}$.
\begin{example}
	Let us consider $T = \mathtt{cbaaa abccba abba}$ in Example~\ref{ex:toy} as an input.
	The shortest reducible prefix of $T$ is $u_1 = \mtt{cbaaa}$, whose suffix $\mtt{aaa}$ is a Z-shape.
	By reducing the string, we obtain $\widehat{u}_1 = \mtt{cba}$.
	By adding letters from the remaining of the input string $T$, it becomes $u_2 = \mtt{cba abc cba}$, which itself is a Z-shape and reduced to $ \widehat{u}_2 = \mtt{cba}$.
	Reading further letters of $T$ gives $u_3 = \mtt{cbaabba}$, which shall be reduced to $\widehat{u}_3 = \mtt{cba}$.
	This is the Z-normal form $\widehat{T}$ of $T$.
\end{example}

\subsection{Z-shape detection}
We first discuss how to find a Z-shape in a pp-irreducible string.
\begin{lemma}\label{lem:uniqueZ}
	Every ss-reducible string has a unique nonempty suffix palindrome
	and thus has a unique Z-shape.
\end{lemma}
\begin{proof}
	Suppose that a string $w$ has a suffix Z-shape occurrence $\lrangle{|w|-2s,|w|-s}$ and a nonempty suffix palindrome centered at $c$ such that $1 \le \rad_w(c) \neq s$.
	It suffices to show that $w$ contains another Z-shape which is not a suffix.
	Figure~\ref{fig:uniqueP} may help understanding the following arguments.
	(a) If $c \le |w|-2s$, then one can find the mirrored occurrence $\lrangle{2c+s-|w|,2c+2s-|w|}$ of that Z-shape with respect to $c$, which is of course not a suffix.
	(b) If $|w|-2s < c < |w|-s$, then $\lrangle{|w|-2s,c}$ is a nonsuffix Z-shape occurrence in $w$.
	(c) If $c > |w|-s$, then one can find the mirrored occurrence of the suffix palindrome with respect to $|w|-s$, 
	whose center is $|w|-2s+\rad_w(c)$.
	That is, $\lrangle{|w|-2s,|w|-2s+\rad_w(c)}$ is a nonsuffix Z-shape occurrence in $w$.
\qed\end{proof}
\begin{figure}[t]
\centering
\begin{tikzpicture}[scale=0.8]\small
\draw(-1, 3) node{\scriptsize(a)};
\zshape[thick]{2.2}{4.4}{1.2}
\zshape{7.6}{9.8}{2.5}
\palin{6}{6}{1.9}
\draw(12, 2.9) node{\scriptsize$|w|$};
\draw(9.8, 2.9) node{\scriptsize$|w|-s$};
\draw(7.6, 2.9) node{\scriptsize$|w|-2s$};
\draw(6, 2.3) node{\scriptsize$c$};
\draw(2.1, 1.6) node{\scriptsize$2c+s-|w|$};
\draw(4.5, 1.6) node{\scriptsize$2c+2s-|w|$};
\draw(-1, 0) node{\scriptsize(b)};
\zshape{4}{8}{-0.5}
\palin{7}{5}{-1.1}
\zshape[thick]{4}{7}{-1.7}
\draw(12, -0.1) node{\scriptsize$|w|$};
\draw(8, -0.1) node{\scriptsize$|w|-s$};
\draw(4, -0.1) node{\scriptsize$|w|-2s$};
\draw(7, -0.7) node{\scriptsize$c$};
\draw(-1, -3) node{\scriptsize(c)};
\zshape{4}{8}{-3.5}
\palin{9}{3}{-4.1}
\zshape[thick]{4}{7}{-4.7}
\draw(12, -3.1) node{\scriptsize$|w|$};
\draw(8, -3.1) node{\scriptsize$|w|-s$};
\draw(4, -3.1) node{\scriptsize$|w|-2s$};
\draw(9, -3.7) node{\scriptsize$c$};
\draw(7, -4.4) node{\scriptsize$|w|-2s+\rad_w(c)$};
\end{tikzpicture}
\caption{\label{fig:uniqueP}%
If $w$ has a suffix Z-shape and two suffix palindromes (drawn with thin lines), it has a non-suffix Z-shape (thick lines).
}
\end{figure}
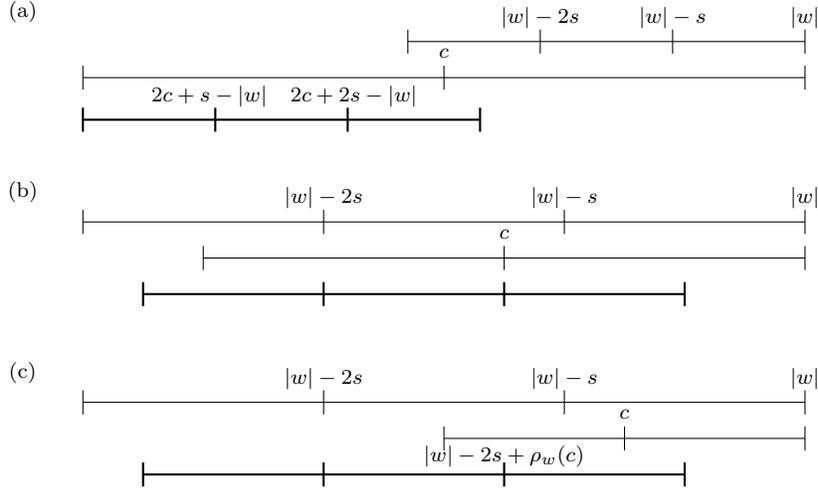
\begin{corollary}\label{cor:uniqueZ}
	Let $c$ be the center of a nonempty suffix palindrome of a pp-irreducible string $w$.
	Then, $w$ is ss-reducible if and only if $\rad_w(c - \rad_w(c)) \ge \rad_w(c)$.
\end{corollary}
There can be several suffix palindromes in an irreducible string.
Lemma~\ref{lem:uniqueZ} implies that only one among those can become\footnote{%
	To avoid lengthy expressions, we casually say that a palindrome centered at $c$ in $x$ \emph{becomes} or \emph{extends to} a (bigger) palindrome in $xy$ when $\rad_x(c) < \rad_{xy}(c)$, without explicitly mentioning several involved mathematical objects that should be understood from the context or that are not important. Other similar phrases should be understood in an appropriate way.} the tail of the unique Z-shape in an ss-reducible string, in which moment the other ones that used to be suffix palindromes are not suffix palindromes any more.
Lemma~\ref{lem:uniqueZ} and Corollary~\ref{cor:uniqueZ} suggest us to keep watching just one (arbitrary) suffix palindrome when reading letters from the input in order to detect a Z-shape.
When the palindrome we are watching has become a non-suffix palindrome, we look for another suffix palindrome to track.
Suppose we are tracking a suffix palindrome centered at $c$ of radius $r = \rad_{w}(c) = |w|-c$ in $w$.
When appending a new letter $t$ from the input to $w$, it is still a suffix palindrome in $wt$ if and only if $wt[c-r]=wt[c+r+1]=t$.
In that case, it is the tail of a Z-shape if and only if $\rad_w(c-r-1) \ge r+1$.

\begin{algorithm2e}[!t]
	\caption{Z-detector} \label{alg:zdet}
	\SetKwFor{Fn}{Function}{}{}
	\SetKwRepeat{Do}{do}{while}
	
	Let $\Pals$ be an empty array and $w = \varepsilon$\;
	\Fn{$\msf{\ZDet}(T)$}{
		Let $T := \texttt{\$} T \texttt{\#}$\tcp*{\texttt{\$} and \texttt{\#} are sentinel symbols}
		$w.\Append(T)$\tcp*{append a new letter from $T$ to $w$}
		\While{there remains to read in $T$}{
			$w.\Append(T)$\;
			$\msf{\ZDetF}(|w|-1)$\;
		}
		\textbf{output} ``No Z-shape'' and \textbf{halt}\;
	}
	\Fn{$\msf{\ZDetF}(c)$}{
		$\msf{Extend}(c)$\;
		\For(\tcp*[f]{in increasing order}){$r \ot 1$ \textbf{to} $\Pals[c]$}{
			\lIf{$r+\Pals[c-r] < \Pals[c]$\label{all:detif}}{%
				$\Pals[c+r] \ot \Pals[c-r]$%
			}\lElse{%
				$\msf{\ZDetF}(c+r)$ and
				\textbf{break}%
			}
		}
	}
	\Fn{$\msf{Extend}(c)$}{  \label{alg:extend}
		$r \ot |w| - c - 1$\;
		\While{$w[c + r + 1] = w[c - r]$}{
			$r \ot r+1$\; 
			\lIf{$\Pals[c - r] \ge r$}{%
				\textbf{ouput} $\lrangle{c-r,c}$ and \textbf{halt}\label{all:find}}
			$w.\Append(T)$\;
		}
		$\Pals[c] \ot r$\;
	}
\end{algorithm2e}
Before presenting our own algorithm for calculating the Z-normal form of an input string, we present an algorithm that finds a Z-shape in an input string following the above tactics.
The algorithm is essentially same as Manacher's algorithm~\cite{Manacher75}, which computes the maximum radius at every position of an input string.
Algorithm~\ref{alg:zdet} outputs the first occurrence of a Z-shape of an input unless it is Z-irreducibe, while computing the maximum radius at each position in an input string.
Commenting out Line~\ref{all:find} gives his original algorithm with slightly different appearance.\footnote{%
Another inessential change from Manacher's algoritm is in Line~\ref{all:detif}.
Our algorithm recurses and breaks when it is confirmed that $r+\Pals[c-r] \ge \Pals[c]$,
while for his original it is only when $r+\Pals[c-r] = \Pals[c]$.
When $r+\Pals[c-r] > \Pals[c]$, his algorithm lets $\Pals[c+r] = \Pals[c] - r$, based on Lemma~\ref{lem:prefast}, and continues iterating the \textbf{for} loop.
}
The algorithm reads letters from the input one by one, while focusing on the leftmost (thus biggest) suffix palindrome.
The algorithm computes the maximum radius at each position from left to right and stores those values in an array $\Pals$.

The function $\msf{Extend}(c)$ calculates $\Pals[c]$ by naively comparing letters on the left and right in the same distance from $c$, knowing that the radius is at least $|w|-c-1$.
That is, the palindrome at $c$ is a suffix palindrome in $\prev{w}=w[1:|w|-1]$ but it is not certain that it is the case in $w$.
When we know that the palindrome at $c$ cannot be extended any further, i.e., that the palindrome at $c$ is not a suffix any more,
we have $w[c+r+1]=w[c-r]$ for all $0 \le r < \rad_w(c)$ and $w[c+r+\rad_w(c)] \neq w[c-\rad_w(c)]$, where $|w|=c+\rad_w(c)+1$. 
Due to the symmetry, the maximum radius at a position $c+r$ in the right arm of a suffix palindrome at $c$ coincides the one at the corresponding position $c-r$ in the left arm for $r < \rad_w(c)$, 
as long as the left end of the palindrome at $c-r$ does not reach the left end of the palindrome at $c$. 
The function $\msf{\ZDetF}(c)$ copies the value of $\Pals[c-r]$ to $\Pals[c+r]$ for $r \le \Pals[c]$ as long as $r+\Pals[c-r] < \Pals[c]$.
If $r+\Pals[c-r] \ge \Pals[c]$, it is not necessarily the case that $\rad_w(c+r) = \rad_w(c-r)$, so we call $\ZDetF(c+r)$ to calculate the radius at $c+r$.
At that time, still we know that $\rad_w(c+r) \ge \rad_w(c)-r = |w|-(c+r)-1$.
That is, $c+r$ is a suffix palindrome center in $\prev{w}$ and is a candidate of a suffix palindrome center in $w$.
Thus the function $\msf{Extend}(c+r)$ starts comparison of letters on the ends of both arms of the palindrome at $c+r$.

By the correctness of Manacher's algorithm and Corollary~\ref{cor:uniqueZ}, we see that Algorithm~\ref{alg:zdet} outputs the Z-shape occurrence of the shortest ss-reducible prefix of the input.
If the input has no Z-shape, it halts with the array $\Pals$ such that $\Pals[c]=\rad_w(c)$ for all the positions $c$.

\subsection{Palindrome chain and stable positions}
A nice property of Algorithm~\ref{alg:zdet} is that when extending a suffix palindrome at $c$, for all positions $d < c$, we have already computed $\Pals[d] = \rad_w(d)$ so that the Z-shape with right pivot $c$ can be detected certainly (if one exists).
One may think of using Algorithm~\ref{alg:zdet} to compute the Z-normal form by deleting the tail of the found Z-shape.
However, deleting a Z-shape tail alters the already calculated maximal radii even on positions that are not deleted.
Maintaining those values is not a trivial issue. 
The following example demonstrates that it should take more than linear time for Z-normalization if we adhere to keep the nice property of Algorithm~\ref{alg:zdet}.
\begin{example}\label{ex:deletetail}
	Consider input string
	$T_m = v_m \mathtt{a}^{2^m}$ where
	$v_0 = \mathtt{ba}$ and 
	$v_i = v_{i-1}^R \mtt{a} t_i t_i \mtt{a} v_{i-1}$ where $t_i$ is a letter not in $v_{i-1}$ for each $i \ge 1$.
	Note that $v_i^R = v_i$ unless $i=0$ and $\mtt{a} v_i$ is a suffix of $v_j$ for all $i <j$.
	The length $n$ of $T_m$ is $\mrm{O}(2^m)$. 
	For example, $v_2 = \mathtt{aba11aba a22a aba11aba}$ where $t_1 = \mathtt{1}$ and $t_2 = \mathtt{2}$.	
	Here $v_m \mtt{aa}$ is an ss-reducible string that has a suffix Z-shape whose tail is $\mtt{aa}$,
	which is the unique nonempty suffix palindrome by Lemma~\ref{lem:uniqueZ}.
	If we provide Algorithm~\ref{alg:zdet} with $v_m \mtt{aa}$, it calculates $\Pals[i]=\rad_{v_m\mtt{aa}}(i)$ for all $i \le |v_m|$
	before detecting the suffix Z-shape by $\Pals[|v_m \mtt{a}|] = \Pals[|v_m|] = 1$.
	The Z-normal form of $v_m \mtt{aa}$ is $v_m$.
	The irreducible string $v_m \mtt{a}$ has $m$ suffix palindromes;
	$\mtt{aa}$ and $\mtt{a}v_i\mtt{a}$ for each $i=1,\dots,m-1$.
	Let $c_i$ be the center of the suffix occurrence of the palindrome $\mtt{a}v_i\mtt{a}$ for $1 \le i < m$.
	Here $\rad_{v_m \mtt{aa}}(c_i) =\rad_{v_m \mtt{a}}(c_i) = \rad_{v_m}(c_i) + 1$.
	Therefore, if we maintain the maximum radius of every palindrome each time we contract a Z-shape, it takes at least $\Omega(m \times 2^m) = \Omega(n \log n)$ time to get the Z-normal form $v_m$ of $T_m$.
\end{example}
Therefore, we have to partly give up to maintain the exact values of maximal radii.
However, under a certain condition, maximal radii become stable and any appended string will not alter the values.
This subsection introduces key technical notions and discusses the condition for positions to be stable.

For distinct positions $c$ and $d$ in an pp-irreducible string $w$,
let us write $c \plink_w d$ if $c \le d-\rad_w(d) \le d \le c+\rad_w(c) \le d+\rad_w(d)$.
Clearly $c \plink_w d$ implies $\rad_w(c) \ge 1$. 
If $c = d-\rad_w(d)$, then $\lrangle{c,d}$ is a Z-shape occurrence in $w$.
Actually the condition $c \le d-\rad_w(d)$ in the above definition is redundant for a pp-irreducible string:
one can see that if $d-\rad_w(d) < c < d \le c+\rad_w(c) \le d+\rad_w(d)$, then $\lrangle{c,d}$ is a non-suffix Z-shape.
Since this paper discusses only pp-irreducible strings, to claim $c \plink_w d$, it is enough to confirm that $c < d \le c+\rad_w(c) \le d+\rad_w(d)$.
A \emph{palindrome chain from $c_0$ in $w$} is a sequence $\vec{c}=(c_0,\dots,c_k)$ of positions in $w$ such that $c_{i-1} \plink_w c_{i}$ for each $i \in [1 : k]$.
We have $\rad_w(c_i) \ge 1$ for each $i \in [0:k-1]$.
The \emph{frontier of the palindrome chain $\vec{c}$ in $w$} is the position $F_w(\vec{c}) = c_k + \rad_w(c_k)$,
and the \emph{maximum frontier from a position $c$} is
\[
\mscr{F}_w(c) = \max\{\, F_w(\vec{c}) \mid \text{$\vec{c}$ is a palindrome chain from $c$}\,\}
\,.\]
A palindrome chain from $c$ is \emph{maximal} if its frontier is $\mscr{F}_w(c)$.
The \emph{originator} $\mscr{A}_w(d)$ of a position $d$ in $w$ is the smallest position $\mscr{A}_w(d)=c$ such that $c \le d \le \mscr{F}_w(c)$.
\begin{figure}[t]
\centering
\begin{tikzpicture}[scale=0.8]\small
\palin{3.5}{3.5}{1.5}
\palin{7}{3}{1}
\palin{9.5}{2}{0.5}
\palin{11}{1}{0}
\foreach \x/\s in {-1/x,0/a,1/b,2/b,3/c,4/d,5/d,6/e,7/e,8/d,9/d,10/c,11/b,12/b,13/a,14/a,15/b,16/b,17/c,18/d,19/d,20/c,21/b,22/b,23/c,24/y}
	\draw(\x*0.5, 2) node[anchor=west]{\large$\mathstrut\mathtt{\s}$};
\foreach \x in {1,...,26}
	\draw(\x*0.5-0.7, 2.5) node{\scriptsize$\x$};
\end{tikzpicture}
\caption{\label{fig:pchain}%
A string $w = \mathtt{xabbcddeeddcbbaabbcddcbbcy}$ has a palindrome chain $(8,15,20,23)$, whose frontier is $25$.
The originator of any position between 8 and 25 is 8.}
\end{figure}
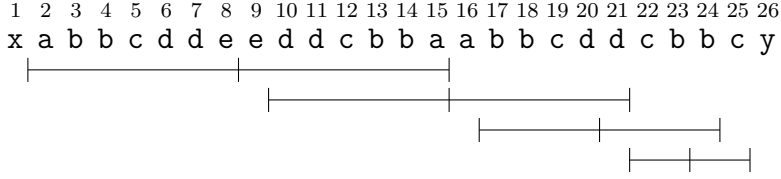
Figure~\ref{fig:pchain} illustrates a palindrome chain in a string $w = \mathtt{xabbcddeeddcbbaabbcddcbbcy}$.                                                                               
\begin{observation}\label{obs:frontier}
Every pp-irreducible string $w$ is uniquely factorized as $w = w_1 \dots w_k$ so that $\mscr{F}_w(|w_1 \dots w_{i-1}|+1)=|w_1 \dots w_{i}|$ for each $i \in [1:k]$
 and $\mscr{A}_w(c) = |w_1 \dots w_{i-1}|+1$ if $|w_1 \dots w_{i-1}| < c \le |w_1 \dots w_{i}| $.
\end{observation}
\begin{defn}\label{def:stable}
	Let $w$ be a pp-irreducible string and $c$ a position in $w$.
	We say that $c$ is \emph{stable} in $w$, if $\mscr{F}_w(c) < |w|$.
	Moreover, $c$ is \emph{strongly stable} if all positions in $[1: c]$ are stable.
\end{defn}

\begin{lemma}\label{lem:stable}
If $c$ is stable in a pp-irreducible string $w$, then $|\widehat{wx}|>\mscr{F}_w(c)$ and $w[c:\mscr{F}_w(c)+1]=\widehat{wx}[c:\mscr{F}_w(c)+1]$ for any string $x$, unless $x$ has a prefix $y$ such that $|\widehat{wy}| < c$.
\end{lemma}
\begin{proof}
We can show the lemma by induction on $|x|$.
For $x=\varepsilon$, suppose $|\widehat{w}| \ge c$ but the claim does not hold.
Then $w$ must have a suffix Z-shape $\lrangle{p,(p+|w|)/2}$ such that $c \le |\widehat{w}| = p < \mscr{F}_w(c)$.
Let $(c_0,\dots,c_k)$ be a maximal palindrome chain where $c=c_0$ and $F_w(c_0,\dots,c_k)=\mscr{F}_w(c)$.
If $p > c_k$, then $(c_0,\dots,c_k,(p+|w|)/2)$ is a palindrome whose frontier is $|w|$, which is a contradiction.
Otherwise there exists $i$ such that $c_i < p \le c_{i+1}$.
Since $\lrangle{c_i,p}$ is not a Z-shape, $(c_0,\dots,c_i,p,(p+|w|)/2)$ is a palindrome chain whose frontier is $|w|$.

For $x=vt$ with $v \in \Sigma^*$ and $t \in \Sigma$, let $u=\widehat{wv}$.
By the induction hypothesis, $w[c:\mscr{F}_w(c)]=u[c:\mscr{F}_w(c)]$, which implies $w[c:\mscr{F}_w(c)]=ut[c:\mscr{F}_w(c)]$, unless $v$ has a prefix $y$ such that $|\widehat{wy}| < c$.
The same argument as the base case applies to $ut$, which completes the proof.
\qed\end{proof}
Lemma~\ref{lem:stable} implies that if $c$ is stable in $w$, then it remains stable in $\widehat{wx}$ unless $x$ has a prefix $y$ such that $|\widehat{wy}| < c$.
\begin{corollary}\label{cor:stronglystable}
If $c$ is strongly stable in $w$, then
 we have $|\widehat{wx}| > \mscr{F}_w(c)$
 and $c$ is strongly stable in $\widehat{wx}$ for any $x$.
\end{corollary}
\begin{proof}
It suffices to show that there is no $y$ such that $|\widehat{wy}| = d$ for some $d < c$.
Since all positions $d \in [1: c]$ are stable, Lemma~\ref{lem:stable} implies that there is no $y$ such that $|\widehat{wy}| = d$.
\qed\end{proof}

\begin{observation}\label{obs:unstable}
If $c$ is not stable, then there is $x$ such that $|\widehat{wx}|=c$ unless $|\widehat{w}| < c$.
\end{observation}
\begin{proof}
Suppose $w$ is irreducible.
If $|w|$ = c, then we have done.
Otherwise, let $(c_0,\dots,c_k)$ be a maximal palindrome chain where $c=c_0$ and $F_w(c_0,\dots,c_k)=\mscr{F}_w(c)=|w|$.
Without loss of generality, we may assume $\rad_w(c_k) \neq 0$.
Then for $x_k = w[2c_k-|w|+1:c_k]$, $wx_k$ has a suffix Z-shape $x_k \rev{x_k} x_k$ whose Z-normal form is $\widehat{wx_k} = w[1:c_k]$.
That is, $F_{\widehat{wx_k}}(c_0,\dots,c_{k-1})=|\widehat{wx_k}| \ge c$.
By repeatedly applying the argument, we obtain $x = x_k \dots x_0$ for which $\widehat{wx} = w[1:|c|]$.

If $w$ is ss-reducible, let $(c_0,\dots,c_k)$ be a maximal palindrome chain such that $c=c_0$ and $\lrangle{c_{k-1},c_k}$ is the suffix Z-shape.
An argument similar to the one in the proof of Lemma~\ref{lem:stable} shows that indeed there must be such a palindrome chain unless $|\widehat{w}|<c$.
Here $(c_0,\dots,c_{k-1})$ is a maximal palindrome chain in $\widehat{w}$ whose frontier is $|\widehat{w}|$.
Then the argument in the previous paragraph applies to $\widehat{w}$.
\qed\end{proof}
In the proof of Observation~\ref{obs:unstable}, we see that $c$ is the center of a nonempty suffix palindrome in $\widehat{wx'}$ for $x'=x_k \dots x_1$.
That is, depending on the letters following $x'$, the maximum radius at $c$ changes, even when the palindrome at $c$ has not been involved in Z-shape reductions.
In Example~\ref{ex:deletetail}, the positions $c_1,\dots,c_m$ are all unstable in (prefixes of) $v_m \mtt{aa}$.
To realize a linear time Z-reduction algorithm, we must avoid recalculation of the maximum radii at those unstable positions as much as possible.

Before presenting our algorithm for Z-reduction in the next subsection, we introduce some technical lemmas below.
\begin{lemma}\label{lem:pchain}
If $c < d \le c+\rad_w(c)$, then $\mscr{F}_w(c) \ge \mscr{F}_w(d)$.
\end{lemma}
\begin{proof}
Let $(d_0,\dots,d_k)$ be a maximal palindrome chain from $d=d_0$, whose frontier is $\mscr{F}_w(d)$.
Suppose $\mscr{F}_w(c) < \mscr{F}(d)$.
There must be $i$ such that $d_i \le c+\rad_w(c) < d_{i}+\rad_w(d_i)$, for which $c \plink_w d_i$.
That is, $(c,d_i,\dots,d_k)$ is a palindrome chain from $c$, whose frontier is $\mscr{F}_w(d) > \mscr{F}_w(c) $.
A contradiction.
\qed\end{proof}
The following strengthens Lemma~\ref{lem:pchain}.
\begin{corollary}\label{cor:pchain}
If $c < d \le \mscr{F}_w(c)$, then $\mscr{F}_w(c) \ge \mscr{F}_w(d)$.
\end{corollary}
\begin{proof}
Let $(c_0,\dots,c_k)$ be a maximal palindrome chain from $c=c_0$, whose frontier is $\mscr{F}_w(c)$.
Then, there must be $i$ for which $c_i < d \le c_i+\rad_w(c_i)$ holds.
Lemma~\ref{lem:pchain} implies $\mscr{F}_w(c) = \mscr{F}_w(c_i) \ge \mscr{F}_w(d)$.
\qed\end{proof}
\begin{corollary}\label{cor:pchain2}
For two positions $b$ and $c$ in $w$ with $b \le c$, let $f = \max\{\, \mscr{F}_w(d) \mid b \le d \le c \,\}$.
Then, we have
$
	f = \max\{\, \mscr{F}_w(d) \mid b \le d \le f\,\} 
$.
\end{corollary}
\begin{proof}
Let $d \in [b:c]$ be such that $\mscr{F}_w(d)=f$.
Then for an arbitrary $e \in [c+1:f] \subseteq [d+1 : \mscr{F}(d)]$, Corollary~\ref{cor:pchain} implies $\mscr{F}_w(e) \le \mscr{F}_w(d) = f$.
\qed\end{proof}
The following lemma has been observed by Manacher~\cite{Manacher75}.
\begin{lemma}\label{lem:prefast}
	For any position $c$ in any string $w$ and any $r \in [1:\rad_w(c)]$,
	$\rad_w(c) \ge r+\rad_w(c+r)$ if $\rad_w(c+r)=\rad_w(c-r)$.
	Otherwise, $\rad_w(c) =  r + \min\{ \rad_w(c + r), \rad_w(c - r)\}$.
\end{lemma}

Using the notion of palindrome chains, one can observe the following property on the behavior of Algorithm~\ref{alg:zdet} for an irreducible string $T$.
Let $c_{0,0} = 1$ and $c_{i,j+1}$ be the leftmost position such that $c_{i,j} \plink_w c_{i,j+1}$. If $c_{i,j}$ has no $d$ such that $c_{i,j} \plink_w d$, where $\rad_w(c_{i,j})=0$, let $c_{i+1,0} = c_{i,j}+1$.
Lemma~\ref{lem:pchain} implies that $(c_{i,0},\dots,c_{i,j})$ forms a maximal palindrome chain from $c_{i,0}$ and $c_{i+1,0} = \mscr{F}_w(c_{i,0}) + 1$.
The function ${\ZDet}(T)$ calls $\ZDetF(c_{i,0})$ for each $i$ and then $\ZDetF(c_{i,j})$ recursively calls $\ZDetF(c_{i,j+1})$. 

\subsection{Outline of our algorithm}

\begin{algorithm2e}[!tp]
	\caption{Z-reducer}\label{alg:zred}
	\SetKwFor{Fn}{Function}{}{}
	\SetKwRepeat{Do}{do}{while}
	Let $\Stack$ be an empty stack, $\Pals$ an empty array and $w = \varepsilon$\;
	\Fn{$\ZRed(T)$}{  \label{alg:LAmanacher}
		$T := \texttt{\$} T \texttt{\#}$\tcp*{\texttt{\$} and \texttt{\#} are sentinel letters}
		$w.\Append(T)$\;
		\While{there remains to read in $T$}{
			$w.\Append(T)$\;
			$\Stack.\mit{clear}()$\;
			$\PCE(|w|-1)$\;
		}
		\Return $w[2:|w|-1]$\tcp*{strip the sentinel letters}
	}
	\Fn{$\PCE(c)$}{  \label{alg:lamanacher}
		$b \ot |w|$\;
		$\unstable \ot \msf{true}$\;
		\While{$\unstable$}{
			$\unstable \ot \msf{false}$\;
			\lIf{$\slowExt(c)$}{\Return{$\msf{true}$}} \label{al:true1}
			\For(\tcp*[f]{in decreasing order}){$d \ot c+\Pals[c]$ \textbf{downto} $b$}{\label{al:for}
				\If{$d+\Pals[d] \ge c+\Pals[c]$\label{al:if}}{
					\If{$\fastExt(d)$}{
						\If{$\PCE(d)$}{
							\lIf{$c = |w|$}{\Return $\msf{true}$} \label{al:true2}
							\lIf{$d = |w|$}{$\Pals[d] := \Pals[2c-d]$}
							$w.\Append(T)$\; \label{al:append}
							$\unstable \ot \msf{true}$\;
							\textbf{break}\;\label{al:break}
						}
					}
					$\Stack.\mit{push}(d)$\;
				}
			}
		}
		\Return $\msf{false}$\;
	}
\end{algorithm2e}
\begin{algorithm2e}[t]
	\caption{Slow and Fast Extension} \label{alg:zext}
	\SetKwFor{Fn}{Function}{}{}
	\SetKwRepeat{Do}{do}{while}
	\Fn{$\slowExt(c)$}{  \label{alg:slowextend}
		$r \ot |w| - c - 1$\;
		\While{$w[c + r + 1] = w[c - r]$}{
			$r \ot r+1$\; 
			\If(\tcp*[f]{detect a suffix Z-shape $\lrangle{c-r,c}$}){
				$\Pals[c - r] \ge r$}{\label{line:detectZ}
				$w \ot w[1:c-r]$\tcp*{contract the suffix Z-shape}
				$\Pals \ot \Pals[1:c-r]$\tcp*{\qquad same as above \qquad\quad }
				\Return $\msf{true}$\;
			}
			$\Pals[c + r] \ot \Pals[c - r]$\tcp*{transfer the value}
			$w.\Append(T)$\; \label{al:append2}
		}
		$\Pals[c] \ot r$\tcp*{$\Pals[c] = \rad_w(c)$}
		\Return $\msf{false}$\;
	}
	\Fn{$\fastExt(d)$}{  \label{alg:fastextend}
		\While{$\Stack$ is not empty}{
			$r \ot \Stack.top() - d$\;
			\lIf{$\Pals[d-r] \ge \Pals[d+r]$}{$\Stack.pop()$}
			\Else(\tcp*[f]{$\Pals[d] = \rad_w(d)$}){
				$\Pals[d] \ot r+\Pals[d-r]$\;
				\Return{$\msf{false}$}
			}
		}
		\Return $\msf{true}$\;
	}
\end{algorithm2e}

Our online algorithm for calculating the Z-normal form of an input string $T$ is shown as Algorithms~\ref{alg:zred} and~\ref{alg:zext}.
Throughout the algorithm, the string $w$ in the working space is kept pp-irreducible.
That is, $\prev{w} = w[1:|w|-1]$ is irreducible and we would like to know if $w$ itself is still irreducible.
Our algorithm consists of functions $\PCE$, $\slowExt$ and $\fastExt$ in addition to the main function $\ZRed$.
Among those, $\PCE$ plays the central role.
The data structures we use are very simple:
a working string $w$, an array $\Pals$ for the maximal radius at each position of $w$, and a stack $\Stack$ of positions.
Those are all global variables in Algorithms~\ref{alg:zred} and~\ref{alg:zext}.
At the beginning, we add extra fresh letters $\texttt{\$}$ and $\texttt{\#}$ to the left and right ends of the input, respectively.
Those work as sentinels so that we never try to access the working string beyond the ends when extending a suffix palindrome.

The working string is initialized to be the empty string and is expanded by appending letters from $T$ one by one by $\Append$.
Suppose that we have a pp-irreducible string $w$ in the working space.
When the function $\PCE(c)$ is called, we know that $c + \rad_{\prev{w}}(c) = |\prev{w}|$, i.e., $c$ is a suffix palindrome center in $\prev{w}$, but not yet sure if $c + \rad_w(c) = |{w}|$ holds, i.e., $c$ may not be a suffix palindrome center in ${w}$.
For explanatory convenience, let us first assume that $w$ will not become an ss-reducible string whose Z-shape includes the position $c$ in its tail.
We will explain later what happens when the position $c$ shall be deleted.
That is, the position $c$ will not be deleted.
Then $\PCE(c)$ processes the shortest prefix $v$ of the unprocessed suffix of $T$ such that $c$ is stable in the resultant string $w'$, i.e., $\mscr{F}_{w'}(c) = |w'|-1$, where  $w'$ is a pp-irreducible string obtained from $wv$ by contracting suffix Z-shapes whose right pivot is right to $c$.
In an extreme case, we have $w'=w$ and just confirm $\mscr{F}_{w}(c) = |{w}|-1$.
After the execution of $\PCE(c)$, unless it returns $\msf{true}$, it is guaranteed that all positions $d \in  [c : \mscr{F}_{w'}(c)]$ are stable in $w'$ and satisfy $\Pals[d] = \rad_{w'}(d)$ for all $d \in [c : \mscr{F}_{w'}(c)]$.
This is why we name the function $\PCE$. 
Moreover if the call of $\PCE(c)$ has been from the main function $\ZRed(T)$, $c = \mscr{A}_{w'}(c)$ and those positions $d$ are all strongly stable. 

To stabilize all the positions up to the (future) frontier of $c$, $\PCE(c)$ recursively calls $\PCE(d)$ for positions $d$ such that $c \plink_w d$.
This accords with the definition of the frontier.
To determine positions $d$ for which we should recursively call $\PCE(d)$, we need to know the value of $\rad_w(c)$ first of all.
The function $\PCE(c)$ calls $\slowExt(c)$ at first.
When the function $\slowExt(c)$ is called, we are sure $c+\rad_{\prev{w}}(c) = |\prev{w}|$.
By reading more letters from the input, it does three tasks.
One is to calculate the maximal radius at $c$ exactly, taking the unread part of the input into account.
One is to detect and contract a suffix Z-shape whose right pivot is $c$.
The last one is to transfer the values of $\Pals$ on the left arm to the right arm.
We extend the palindrome at $c$ by comparing the letters $w[c-r]$ and $w[c+r+1]$.
When it happens that $\Pals[c-r] \ge r$, this means that we find a Z-shape occurrence $\lrangle{c-r,c}$.
In this case, the suffix palindrome shall be deleted, and the function returns $\msf{true}$.
When the palindrome has become non-suffix, it returns $\msf{false}$.
During the extension of the palindrome at $c$, it copies the value of $\Pals[c-r]$ to $\Pals[c+r]$.
This transfer might appear nonsense, since it might be the case that $\rad_w(c-r) \neq \rad_w(c+r)$.
However, this ``sloppy calculation'' of radii is advantageous over the exactly correctly calculated values.
The copied value at $c+r$ is ``adaptive'' in extensions and deletions of succeeding part of the working string $w$ to some extent, in the sense that they can always be used to detect a Z-shape occurrence $\lrangle{c+r,d}$ as long as $d \le c+\rho_w(c)$.
The exactly correct values are too rigid to have this property.
If $\Pals[c+r]=\rad_w(c+r)$, of course $\lrangle{c+r,d}$ is a Z-shape if and only if $\Pals[c+r],\rad_w(d)  \ge d-(c+r)$.
It is possible that $\Pals[c+r] \neq \rad_w(c+r)$, but still $\lrangle{c+r,d}$ is a Z-shape if and only if $\Pals[c+r],\rad_w(d) \ge d-(c+r)$, as long as $d \le c+\rho_w(c)$.
Consider the situation illustrated in Figure~\ref{fig:adaptive}.
In the case where $\Pals[c-r] = \Pals[c+r] \neq \rad_w(c+r)$, it is certain that $\rad_w(c+r),\Pals[c+r] \ge \rad_w(c) - r$.
If $c+r < d \le c+\rad_w(c)$, it means $\rad_w(c+r) \ge d-(c-r)$.
Thus $\lrangle{c+r,d}$ is a Z-shape if and only if $\Pals[c+r],\rad_w(d) \ge d-(c+r)$.
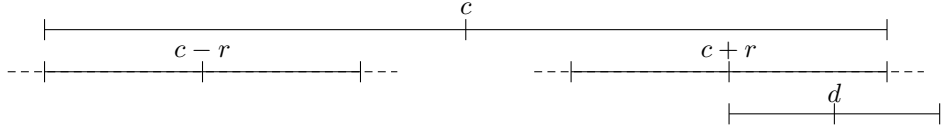
\begin{figure}[t]
	\centering
	\begin{tikzpicture}[scale=0.7]\small
	\draw(8,0.4) node{$c$};
	\palin{8}{8}{0}
	\draw(13,-0.4) node{$c+r$};
	\palin{13}{3}{-0.8}
	\draw[densely dashed](9.3,-0.8)--(16.7,-0.8);
	\draw(3,-0.4) node{$c-r$};
	\palin{3}{3}{-0.8}
	\draw[densely dashed](-0.7,-0.8)--(6.7,-0.8);
	\draw(15,-1.2) node{$d$};
	\palin{15}{2}{-1.6}
	\end{tikzpicture}
	\caption{\label{fig:adaptive}%
	If $\Pals[c+r] \neq \rad_w(c+r)$, we know $\rad_w(c+r),\Pals[c+r] \ge \rad_w(c) - r$.
	This is informative enough for detecting a Z-shape whose left pivot is $c+r$ and right pivot is $d \le c+\rad_w(c)$.
	}
\end{figure}
The following example shows how those copied values work well.
\begin{example}\label{ex:transfer}
	See Figure~\ref{fig:transfer}.
	Let us consider the ss-reducible string $w_1=\mtt{\texttt{\$}caabbaacbbcaabbaab}$ with suffix Z-shape $\mtt{abbaab}$.
	Here we have a big palindrome centered at $10$ whose radius is $8$, i.e., $\rad_{w_1}(10)=8$.
	On the symmetric positions $5$ and $15$ with respect to that palindrome,
	we have $\rad_{w_1}(5) = 4 \neq \rad_{w_1}(15)=3$.
	In $\slowExt(10)$, we transfer the value $\Pals[5]=4$ to $\Pals[15]$.
	The palindrome suffix in $w_1$ is centered at $17$.
	Although $\Pals[15] \neq \rad_{w_1}(15)$, still it is useful to detect the Z-shape occurrence $\lrangle{15,17}$, since $\Pals[15] \ge 17 - 15 = 2 = \rad_{w_1}(17)$.
	After the contraction of the tail $\mtt{baab}$, we obtain $w_2 = \hat{w}_1 = \mtt{\texttt{\$}caabbaacbbcaab}$.
	Suppose we further read $v = \mtt{baaccaab}$.
	Then $w_3 = w_2 v$ is ss-reducible, where the suffix palindrome is centered at $19$ and $\lrangle{15,19}$ is the suffix Z-shape occurrence.
	To detect it, we should know $\rad_{w_3}(15) \ge 19-15 = 4 = \rad_{w_3}(19)$.
	Since we set $\Pals[15]=4$, we can detect the Z-shape, without updating the value.
\begin{figure}[t]
	\centering
	\begin{tikzpicture}[scale=0.75]\small
	\palin{4}{4}{1.5}
	\palin{6.5}{1.5}{1}
	\palin{1.5}{2}{1}
	\zshape{6.5}{7.5}{0.5}
	\foreach \x in {1,...,22}
		\draw(\x*0.5-1.2, 2.5) node{\scriptsize$\x$};
	\draw(-1.5,2) node{$w_1$};
	\foreach \x/\s in {-2/\texttt{\$},-1/c,0/a,1/a,2/b,3/b,4/a,5/a,6/c,7/b,8/b,9/c,10/a,11/a,12/b,13/b,14/a,15/a,16/b}
		\draw(\x*0.5, 2) node[anchor=west]{\normalfont$\mathstrut\mathtt{\s}$};
	\draw(-1.5,0) node{\small$\rad_{w_1}$};
	\foreach \x/\s in {0/0,1/0,2/1,3/0,4/4,5/0,6/1,7/0,8/0,9/8,10/0,11/0,12/1,13/0,14/3,15/0,16/2,17/0,18/{}}
		\draw(\x*0.5-0.55, 0) node{$\s$};
	\draw(-1.5,-0.5) node{\small$\Pals_1$};
	\foreach \x/\s in {0/0,1/0,2/1,3/0,4/4,5/0,6/1,7/0,8/0,9/8,10/0,11/0,12/1,13/0,14/4,15/0,16/\mathbf{2},17/0,18/{}}
		\draw(\x*0.5-0.55, -0.5) node{$\s$};
	\draw(-1.5,-1) node{$w_2$};
	\foreach \x/\s in {-2/\texttt{\$},-1/c,0/a,1/a,2/b,3/b,4/a,5/a,6/c,7/b,8/b,9/c,10/a,11/a,12/b}
		\draw(\x*0.5, -1) node[anchor=west]{\normalfont$\mathstrut\mathtt{\s}$};
	\draw(-1.5,-1.5) node{\small$\Pals_3$};
	\foreach \x/\s in {0/0,1/0,2/1,3/0,4/4,5/0,6/1,7/0,8/0,9/8,10/0,11/0,12/1,13/0,14/4,15/0,16/1,17/0,18/\mathbf{4},19/0,20/1,21/0}
		\draw(\x*0.5-0.55, -1.5) node{$\s$};
	\draw(-1.5,-2) node{\small$\rad_{w_3}$};
	\foreach \x/\s in {0/0,1/0,2/1,3/0,4/4,5/0,6/1,7/0,8/0,9/8,10/0,11/0,12/1,13/0,14/4,15/0,16/1,17/0,18/{4},19/0,20/1,21/0}
		\draw(\x*0.5-0.55, -2) node{$\s$};
	\zshape{6.5}{8.5}{-2.5}
	\palin{4}{4.5}{-3}
	\draw(-1.5,-3.5) node{$w_3$};
	\foreach \x/\s in {-2/\texttt{\$},-1/c,0/a,1/a,2/b,3/b,4/a,5/a,6/c,7/b,8/b,9/c,10/a,11/a,12/b,13/b,14/a,15/a,16/c,17/c,18/a,19/a,20/b}
		\draw(\x*0.5, -3.5) node[anchor=west]{\normalfont$\mathstrut\mathtt{\s}$};
	\end{tikzpicture}
	\caption{\label{fig:transfer}%
	We copy the value $\Pals[5]=4$ to $\Pals[15]$ with regardless of whether or not $\rho(15)=4$, where $5$ and $15$ are symmetric positions with respect to the palindrome center $10$.
	Using the value, we can detect Z-shapes $\lrangle{15,17}$ in $w_1$ and $\lrangle{15,19}$ in $w_3$ with no update of $\Pals[15]$.
	$\Pals_1$ and $\Pals_3$ show the arrays when we have $w_1$ and $w_3$ as a working copy, respectively,
	except for the values with bold letters, which are computed but not written, due to the Z-shape contract.}
\end{figure}
\end{example}

Those values $\Pals[c+r]$ copied from $\Pals[c-r]$ work well to detect a Z-shape $\lrangle{c+r,d}$ only when $d \le c+\rad_w(c)$.
For other positions $d > c+\rad_w(c)$, we may overlook or erroneously report a Z-shape if we leave the values wrong.
In the \textbf{for} loop of Algorithm~\ref{alg:zred}, we try to fix the values $\Pals[d]$ to be $\rad_w(d)$, unless $d + \Pals[d] < c + \Pals[c]$ which witnesses $\Pals[d]=\rad_w(d)$, in decreasing order on the right arm of the palindrome at $c$ .
This ``reversed'' order might appear unnatural, but this is also related to the adaptability of values in $\Pals$.
If we fix $\Pals[d]$ to be $\rad_w(d)$ in increasing order, they are not adaptive any more.
In this case, once some suffix of the working string is deleted and then extended, those exact values would become useless.
Contrarily, we calculate $\rad_w(d)$ in the opposite order.
Then the previously copied values of $\Pals$ on the left are adaptive and remain useful, unless they are deleted.

Next we explain what the function $\fastExt$ does.
Suppose that we have two positions $d$ and $e$ such that $c < d < e \le c+ \rad_w(c)$ and $d + \Pals[d],e + \Pals[e] \ge c + \Pals[c]$.
They are now candidates for suffix palindrome centers in $w$.
We first focus on $e$ and then $d$.
By appending further letters one by one from $T$ to $w$, we extend the palindrome at $e$ until it becomes non-suffix.
Suppose that the palindrome at $e$ has been expanded using $\slowExt(e)$ and the value $\Pals[e]$ has been fixed, where $e + \Pals[e] = |w| - 1$.
We will then fix $\Pals[d]$.
Figure~\ref{fig:skip} illustrates such a situation.
Recall that we know $d + \rad_w(d) \ge c + \rad_w(c)$ from $d+\Pals_w[d] \ge c + \Pals[c]$.
One naive way to obtain the maximum palindrome at $d$ is just to compare letters $w[c+\Pals[c]+r]$ and $w[2d-c-\Pals[c]-r+1]$ for $r=1,2,\dots$ until it becomes non-suffix, just like we have done for $e$.
However, this means that we reread the same letters that have been read when extending the palindrome at $e$ (illustrated by bold lines in Figure~\ref{fig:skip}).
We must avoid this for realizing linear time computation.
Suppose that the palindrome at $d$ should be extended properly.
Then, since $e$ is in the right arm of the palindrome $d$, i.e., $d < e \le d + \rad_w(d)$, 
one can find a symmetric occurrence at $2d-e$ of (a part of) the palindrome at $e$.
We compare the values $\Pals[2d-e]$ and $\Pals[e]$ at the symmetric positions with respect to $d$.
Thanks to Lemma~\ref{lem:prefast}, either we can obtain the exact value of $\rad_w(d)$ or we learn that $d+\rad_w(d) \ge e + \Pals[e] = |w| - 1$.
In the former case, $\fastExt(d)$ lets $\Pals[d] = \rad_w(d)$ and returns $\msf{false}$.
In the latter case, we know that $d$ is a suffix palindrome center in $\prev{w}$.
$\fastExt(d)$ returns $\msf{true}$ and then we call $\PCE(d)$.
\begin{figure}[t]
	\centering
	\newcommand{\dist}[4]{
	\draw[<-] (#1,#3) -- (#2/2+#1/2-0.25,#3);
	\draw[<-] (#2,#3) -- (#2/2+#1/2+0.25,#3);
	\draw (#2/2+#1/2,#3) node{#4};
	}
	\begin{subfigure}[h]{\textwidth}
	\begin{tikzpicture}[scale=0.75]\small
	\palin{8}{6}{6}
	\palin{4}{2}{7}
	\palin{12}{3}{7}
	\draw (0,8) -- (15.5,8);
	\draw (13,7.8) -- (13,8.2);
	\draw (15.5,7.8) -- (15.5,8.2);
	\draw[very thick] (13,8) -- (14,8);
	\draw (10,6.8) -- (10,7.2);
	\draw (14,6.8) -- (14,7.2);
	\draw[dotted] (2,8.2) -- (2,5.8);
	\draw[dotted] (4,8.2) -- (4,5.8);
	\draw[dotted] (8,8.2) -- (8,5.8);
	\draw[dotted] (12,8.2) -- (12,5.8);
	\draw[densely dotted] (13,8.2) -- (13,5.8);
	\draw[dotted] (14,8.2) -- (14,5.8);
	\draw (15.5,8.5) node{$|w|$};
	\draw (13,8.5) node{$c+\rad_w(c)$};
	\draw (12,7.5) node{$e$};
	\draw (8,6.5) node{$d$};
	\draw (4,7.5) node{$e'$};
	\dist{4}{6}{6.85}{\scriptsize$s$}
	\dist{12}{14}{6.85}{\scriptsize$s$}
	\end{tikzpicture}
	\caption{When $\Pals[e'] \neq \Pals[e]$, we have $\Pals[d]=e-d+s$ for $s=\min\{\Pals[e'],\Pals[e]\}$.}
	\end{subfigure}
	\begin{subfigure}[h]{\textwidth}
	\begin{tikzpicture}[scale=0.75]\small
	\palin{8}{6.5}{1}
	\palin{4}{2.5}{2}
	\palin{12}{2.5}{2}
	\draw (0,3) -- (15,3);
	\draw (13,2.8) -- (13,3.2);
	\draw (15,2.8) -- (15,3.2);
	\draw[very thick] (13,3) -- (14.5,3);
	\draw[densely dashed] (14.5,1) -- (16,1);
	\draw[densely dashed] (1.5,1) -- (0,1);
	\draw[densely dotted] (11,2.2) -- (11,1.8);
	\draw[dotted] (1.5,3.2) -- (1.5,0.8);
	\draw[dotted] (4,3.2) -- (4,0.8);
	\draw[dotted] (8,3.2) -- (8,0.8);
	\draw[dotted] (12,3.2) -- (12,0.8);
	\draw[densely dotted] (13,3.2) -- (13,0.8);
	\draw[dotted] (14.5,3.2) -- (14.5,0.8);
	\draw (15,3.5) node{$|w|$};
	\draw (13,3.5) node{$c+\rad_w(c)$};
	\draw (12,2.5) node{$e$};
	\draw (8,1.5) node{$d$};
	\draw (4,2.5) node{$e'$};
	\end{tikzpicture}
	\caption{When $\Pals[e'] = \Pals[e]$, we have $\Pals[d] \ge e-d+\Pals[e]$. We try to extend the palindrome at $d$ by reading more letters shown by the dashed lines.}
	\end{subfigure}
	\caption{\label{fig:skip}%
	The string represented by the bold line is scanned when computing $\Pals[e]$. 
	For computing $\Pals[d]$, we compare $\Pals[e']$ and $\Pals[e]$ for $e'=2d-e$ rather than scanning the string shown by the bold line. 
	}
\end{figure}
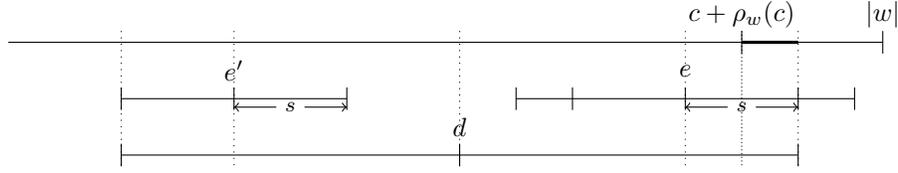
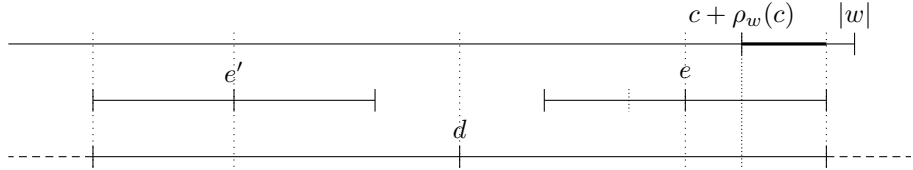

Note that the situation described above is not yet general enough.
To fix $\Pals[e]$, we call $\PCE(e)$, which may call $\PCE(e_1)$ for some $e_1$ such that $e \plink_w e_1$.
The recursive calls from $\PCE(e)$ gives a palindrome chain $(e,e_1,\dots,e_k)$ such that $e_k + \Pals[e_k] = |w| - 1$.
We use $\Stack$ to remember the palindrome chain whose frontier is $|\prev{w}|$ so that 
the function $\fastExt(d)$ can decide whether the palindrome at $d$ is a suffix of $\prev{w}$ by repeatedly applying Lemma~\ref{lem:prefast}.
If the right arm of the palindrome centered at $d$ can reach $|\prev{w}|$, the left arm of it must have the structure that can be seen as the ``reversed palindrome chain'' symmetric to the one in $\Stack$.
By examining whether $\Pals[2d-e_i] = \Pals[e_i]$ for each $i$, one can tell whether the right arm of the maximal palindrome at $d$ can reach the position $|w|-1$.
If it is the case, $\fastExt(d)$ returns $\msf{true}$ and lets $\slowExt(d)$ extend the palindrome.
Otherwise, $\fastExt(d)$ lets $\Pals[d] = \rad_w(d)$ and returns $\msf{false}$.
When $\slowExt(d)$ finds that $c$ is the right pivot of a suffix Z-shape, it returns $\msf{true}$ after deleting the tail.
Then the length of the working string is smaller than $d$ and $\PCE(d)$ can do nothing other than returning $\msf{true}$.
Suppose that $\PCE(d)$ has been called from $\PCE(c)$ for some $c$.
That is, $d$ is in the right arm of the maximum palindrome at $c$ in the string before the Z-shape contraction.
One can see that $|w| \ge c$, since otherwise, $w$ must have had a Z-shape in a proper prefix.
Moreover we have $|w| \le d$, since otherwise $\PCE(|w|)$ was called before $\PCE(d)$ and the Z-shape had been detected before calling $\PCE(d)$.
If $|w|=c$, we lost the precondition for calling $\PCE(c)$, that $c$ is a suffix palindrome center in $\prev{w}$.
So $\PCE(c)$ can do nothing other than returning $\msf{true}$.
If $|w|> c$, it means that the right arm of the palindrome at $c$ is cut in the middle.
Now, $c$ is a suffix palindrome center in $w$.
Thus, we try to extend the palindrome at $c$.
This is how the \textbf{while} loop of $\PCE$ works.
Note that when $|w| = d$, it means that $\PCE(d)$ may have updated the value $\Pals[d]$ from $\Pals[2c-d]$ to the real value $\rad_{w'}(d)$ where $w'$ is the working string when $\PCE(d)$ was called.
In that case we recover the ``adaptive'' value by letting $\Pals[d] = \Pals[2c-d]$.

Here is a running example.
\begin{example}\label{ex:running}
We will explain the algorithm through a running example.
Let us say that a position $d$ is \emph{stabilized} if either
\begin{itemize}
	\item $\PCE(d)$ has returned $\msf{false}$,
	\item the \textbf{if} condition on Line~\ref{al:if} has been confirmed to be $\msf{false}$ for some $c$,
	\item $\fastExt(d)$ has returned $\msf{false}$.
\end{itemize}
In the following explanation, the value of $\Pals[d]$ for stabilized positions $d$ are shown in black.
If $\PCE(d)$ is running, $\Pals[d]$ is in red.
The others are in blue.

Suppose we are given 
\[
T = \mtt{abccbaabbc cbbaaa abccbaabbc}
\,.\]
For the first three letters including the added sentinel letter,  $\ZRed(T)$ computes $\Pals[c]=0$ and returns $\msf{false}$ quickly for $c=1,2,3$.
Then $\ZRed(T)$ calls $\PCE(4)$, which extends the palindrome at $4$ with $\slowExt(4)$ copying the values of $\Pals$ on the left arm to the right arm and lets $\Pals[4]=3$.
Then $\PCE[4]$ calls $\PCE[7]$, which lets $\Pals[7]=2$ and calls $\PCE[9]$, which lets $\Pals[9]=1$ and calls $\PCE[10]$, which lets $\Pals[10]=0$.
We now have
\[	\begin{tikzpicture}[scale=0.7]\small
	\foreach \x in {1,...,25}
		\draw(\x*0.5-0.2, 0.2) node{\scriptsize$\x$};
	\draw(-1,-0.5) node{$w$};
	\foreach \x/\s in {0/\texttt{\$},1/a,2/b,3/c,4/c,5/b,6/a,7/a,8/b,9/b,10/c}
		\draw(\x*0.5, -0.5) node[anchor=west]{\normalfont$\mathstrut\mathtt{\s}$};
	\draw(-1,-1) node{$\Pals$};
	\foreach \x/\s in {1/0,2/0,3/0,4/\snd{3},5/\fst{0},6/\fst{0},7/\snd{2},8/\fst{0},9/\snd{1},10/\snd{0}}
		\draw(\x*0.5-0.05, -1) node{$\s$};
\end{tikzpicture}\]
and	 positions $c=10,9,8,7,6,5,4$ will be stabilized in this order.
At that moment, $\Stack$ has a maximal palindrome chain $(7,9)$, but it will be discarded without playing any important role. 
Now $\ZRed(T)$ calls $\PCE(11)$ after appending one more letter to $w$.
The palindrome at $11$ is extended until it finds $\rad_w(11)=5$ while copying values of $\Pals$ from left to right. 
\[	\begin{tikzpicture}[scale=0.7]\small
	\foreach \x in {1,...,25}
		\draw(\x*0.5-0.2, 0.2) node{\scriptsize$\x$};
	\draw(-1,-0.5) node{$w$};
	\foreach \x/\s in {0/\texttt{\$},1/a,2/b,3/c,4/c,5/b,6/a,7/a,8/b,9/b,10/c,11/c,12/b,13/b,14/a,15/a,16/a}
		\draw(\x*0.5, -0.5) node[anchor=west]{\normalfont$\mathstrut\mathtt{\s}$};
	\draw(-1,-1) node{$\Pals$};
	\foreach \x/\s in {1/0,2/0,3/0,4/{3},5/{0},6/{0},7/{2},8/{0},9/{1},10/{0},11/\snd{5},12/\fst{0},13/\fst{1},14/\fst{0},15/\fst{2},16/\fst{0},17/}
		\draw(\x*0.5-0.05, -1) node{$\s$};
\end{tikzpicture}\]
Here $\Pals[15]=\Pals[7]=2$, though $\rad_w(15)=1$.
Then $\PCE(11)$ calls $\PCE(16)$, where $\slowExt(16)$ realizes $\lrangle{15,16}$ is a Z-shape due to $\Pals[15] \ge \rad_w(16)$ and deletes its tail.
Since $\PCE(16)$ returns $\msf{true}$, $\PCE(11)$ continues extending the palindrome at $11$, up to its maximum radius $10$.
\[	\begin{tikzpicture}[scale=0.7]\small
	\foreach \x in {1,...,25}
		\draw(\x*0.5-0.2, 0.2) node{\scriptsize$\x$};
	\draw(-1,-0.5) node{$w$};
	\foreach \x/\s in {0/\texttt{\$},1/a,2/b,3/c,4/c,5/b,6/a,7/a,8/b,9/b,10/c,11/c,12/b,13/b,14/a,15/a,16/b,17/c,18/c,19/b,20/a,21/a}
		\draw(\x*0.5, -0.5) node[anchor=west]{\normalfont$\mathstrut\mathtt{\s}$};
	\draw(-1,-1) node{$\Pals$};
	\foreach \x/\s in {1/0,2/0,3/0,4/{3},5/{0},6/{0},7/{2},8/{0},9/{1},10/{0},11/\snd{10},12/\fst{0},13/\fst{1},14/\fst{0},15/\fst{2},16/\fst{0},17/\fst{0},18/\fst{3},19/\fst{0},20/\fst{0},21/\fst{0}}
		\draw(\x*0.5-0.05, -1) node{$\s$};
\end{tikzpicture}\]
Then $\PCE(11)$ calls $\PCE(21)$, which calls $\PCE(23)$, which calls $\PCE(24)$.
\[	\begin{tikzpicture}[scale=0.7]\small
	\foreach \x in {1,...,25}
		\draw(\x*0.5-0.2, 0.2) node{\scriptsize$\x$};
	\draw(-1,-0.5) node{$w$};
	\foreach \x/\s in {0/\texttt{\$},1/a,2/b,3/c,4/c,5/b,6/a,7/a,8/b,9/b,10/c,11/c,12/b,13/b,14/a,15/a,16/b,17/c,18/c,19/b,20/a,21/a,22/b,23/b,24/c}
		\draw(\x*0.5, -0.5) node[anchor=west]{\normalfont$\mathstrut\mathtt{\s}$};
	\draw(-1,-1) node{$\Pals$};
	\foreach \x/\s in {1/0,2/0,3/0,4/{3},5/{0},6/{0},7/{2},8/{0},9/{1},10/{0},11/\snd{10},12/\fst{0},13/\fst{1},14/\fst{0},15/\fst{2},16/\fst{0},17/\fst{0},18/\fst{3},19/\fst{0},20/\fst{0},21/\snd{2},22/\fst{0},23/\snd{1},24/\snd{0},25/{}}
		\draw(\x*0.5-0.05, -1) node{$\s$};
\end{tikzpicture}\]
Positions $24,23,22,21,20,19$ will be stabilized and
 $\Stack$ has the maximal palindrome chain $(21,23)$.
 When fixing the value $\Pals[18]$, this plays an important role.
 We have known that $\rad_w(18) \ge 3$.
 Referring to the values of $\Pals$ of positions $(21,23)$ in $\Stack$ and their symmetric positions $(15,13)$ with respect to $18$, $\fastExt(18)$ tells us that $\rad_w(18) \ge 6$, without comparing $w[18+r+1]$ and $w[18-r]$ for $r =4,5,6$, as illustrated below. 
\[	\begin{tikzpicture}[scale=0.7]\small
	\foreach \x in {1,...,25}
		\draw(\x*0.5-0.2, 0.2) node{\scriptsize$\x$};
	\draw(-1,-0.5) node{$w$};
	\foreach \x/\s in {0/\texttt{\$},1/a,2/b,3/c,4/c,5/b,6/a,7/a,8/b,9/b,10/c,11/c,12/b,13/b,14/a,15/a,16/b,17/c,18/c,19/b,20/a,21/a,22/b,23/b,24/c}
		\draw(\x*0.5, -0.5) node[anchor=west]{\normalfont$\mathstrut\mathtt{\s}$};
	\draw(-1,-1) node{$\Pals$};
	\foreach \x/\s in {1/0,2/0,3/0,4/{3},5/{0},6/{0},7/{2},8/{0},9/{1},10/{0},11/\snd{10},12/\fst{0},13/\fst{1},14/\fst{0},15/\fst{2},16/\fst{0},17/\fst{0},18/\fst{3},19/{0},20/{0},21/{2},22/{0},23/{1},24/{0},25/{}}
		\draw(\x*0.5-0.05, -1) node{$\s$};
	\palin{5.5}{5}{-1.5}
	\palin{6.5}{0.5}{-2}
	\palin{7.5}{1}{-2.5}
	\palin{9}{1.5}{-3}
	\palin{10.5}{1}{-2.5}
	\palin{11.5}{0.5}{-2}
	\palin[thick]{9}{3}{-3.5}
\end{tikzpicture}\]
  Then $\PCE(18)$ is called, which properly extends the palindrome at $18$ using $\slowExt(18)$ and learns $\rad_w(18) \ge 7$,
 at which moment the Z-shape $\lrangle{11,18}$ is detected by $\Pals[18-7] \ge 7$.
After contracting the Z-shape, one more letter is appended, which is the last sentinel $\texttt{\#}$.
We get $\rad_w(11)=0$ and the algorithm terminates with $w[2:|w|-1] = \widehat{T} = \mtt{abccbaabbc}$.
\end{example}

\subsection{Correctness and complexity of the algorithm}
To prove the correctness of our algorithm, we first introduce some technical definitions, which characterize ``adaptive'' values.
\begin{defn}
	Let us write $i \sim_k j$ if $\min\{i,k\} = \min\{j,k\}$.
	We say that $\Pals$ on $w$ is \emph{accurate enough between $c$ and $d$} if for any $e \in [c:d]$, it holds that $\Pals[e] \sim_{d-e} \rad_w(e)$.
	This property is denoted by $\ACE_w(c,d)$ with implicit understanding of $\Pals$.

	Let $\nu_w(c) $ denote the largest $e$ such that $e \plink_w c$.
	If there is no such $e$, let $\nu_w(c) = 1$.
	We say that $c$ is \emph{left-good} in $w$ if $\ACE_w(\nu_w(c),c)$ holds.
	We say that $c$ is \emph{right-good} in $w$ if $\ACE_w(c,c+\rad_w(c))$ holds.
\end{defn}
Clearly if $\ACE_w(c_1,d_1)$ and $[c_2:d_2] \subseteq [c_1:d_1]$, then $\ACE_w(c_2,d_2)$ holds. 
\begin{lemma}\label{lem:preslow}
	Suppose that a suffix palindrome center $c$ is left-good in a pp-irreducible string $w$.
	Then $w$ has a Z-shape occurrence $\lrangle{c-\rad_w(c),c}$ if and only if $\Pals[c-\rad_w(c)] \ge \rad_w(c) \ge 1$.
\end{lemma}
\begin{proof}
	Suppose that $\lrangle{c-\rad_w(c),c}$ is a Z-shape in $w$, where $\rad_w(c-\rad_w(c)) \ge \rad_w(c)$.
	For each $r \in [1:\rad_w(c)-1]$, since $\lrangle{c-r,c}$ is not a Z-shape in $w$, $\rad_w(c-r) < r$.
	We have $\nu_{w}(c) = c-\rad_w(c)$ and the left-goodness of $c$ implies $\Pals[c-\rad_w(c)] \ge \rad_w(c)$.
	
	Suppose that $\lrangle{c-\rad_w(c),c}$ is not a Z-shape in $w$,	i.e., $\rad_w(c-\rad_w(c)) < \rad_w(c)$.
	Since $w$ is pp-irreducible,  $\rad_w(c-r) < r$ for any $r \in [1:\rad_w(c)-1]$.
	This means that $\nu_{w}(c) < c-\rad_w(c)$ and the left-goodness of $c$ implies $\Pals[c-\rad_w(c)] = \rad_w(c-\rad_w(c)) < \rad_w(c)$.
\qed\end{proof}
During execution of $\slowExt(c)$, it does not necessarily hold that $\Pals[d]=\rad_w(d)$ for $d <c$,
but $c$ is guaranteed to be left-good as we will show later.
Hence, by Lemma~\ref{lem:preslow}, we certainly detect a Z-shape whose right pivot is $c$.

Since $\slowExt(c)$ copies values of $\Pals[c-r]$ to $\Pals[c+r]$ unless a Z-shape is detected, it may be the case that $\Pals[c+r]\neq\rad_w(c+r)$.
The function $\PCE(c)$ may call $\PCE(d)$ for $d$ such that $c \plink_w d$.
Lemmas~\ref{lem:leftgood} to~\ref{lem:right2left} imply that left-goodness will be inherited from $c$ to $d$.
\begin{lemma}\label{lem:leftgood}
	Suppose that $c$ is left-good in a pp-irreducible string $w$.
	If $c$ is left-good, then $\Pals[c-r]=\rad_w(c-r)$ for all $r \in [1:\rad_{w}(c)-1]$.
	Moreover $\Pals[c-\rad_w(c)]=\rad_w(c-\rad_w(c))$ if $w$ has no Z-shape whose right pivot is $c$.
\end{lemma}
\begin{proof}
	Since $\lrangle{c-r,c}$ is not a Z-shape, we have $\rad_w(c-r) < r$ for all $r \in [1:\rad_w(c)-1]$.
	Thus, $\nu_{w}(c) \le c-\rad_w(c)$.
	The left-goodness of $c$ implies $\rad_w(c-r) = \Pals[c-r] $ for all $r \in [1:\rad_w(c)-1]$.
	Moreover if $\lrangle{c-\rad_w(c),c}$ is not a Z-shape, $\nu_{w}(c) < c-\rad_w(c)$ and thus  $\rad_w(c-\rad_w(c)) = \Pals[c-\rad_w(c)] $.
\qed\end{proof}

\begin{lemma}\label{lem:left2right}
	Suppose $\Pals[c+r]=\rad_w(c-r)$ for all $r \in [1:\rad_{w}(c)]$.
	Then, $c$ is right-good.
\end{lemma}
\begin{proof}
	For $r \in [1:\rad_w(c)]$, suppose that $c-r-\rad_w(c-r) > c - \rad_w(c)$.
	That is, the end of the left arm of the maximum palindrome at $c-r$ is bigger than the left end of that at $c$.
	Due to the symmetry w.r.t.\ the center $c$, $\rad_w(c+r) = \rad_w(c-r) = \Pals[c+r]$ in this case.

	Suppose otherwise, $c-r-\rad_w(c-r) \le c - \rad_w(c)$, i.e., $\rad_w(c-r) \ge \rad_w(c) - r$.
	That is, the left arm of the maximum palindrome at $c-r$ reaches the left end of that at $c$.
	Due to the symmetry w.r.t\ $c$, $c+r+\rad_w(c+r) \ge c + \rad_w(c)$ in this case.
	Here we have $\rad_w(c-r),\rad_w(c+r) \ge \rad_w(c) - r$.
	Hence $\Pals[c+r] \sim_{\rad_w(c)-r} \rad_w(c+r)$ by $\Pals[c+r]=\Pals[c-r]=\rad_w(c-r)$, 

	Therefore, $c$ is right-good.
\qed\end{proof}
\begin{lemma}\label{lem:right2left}
	Suppose that $c$ is right-good and $c \plink_w d$. 
	Then $d$ is left-good. 
\end{lemma}
\begin{proof}
	Since $[\nu_w(d):d] \subseteq [c:d] \subseteq [c:\rad_w(c)]$,
	 $\ACE_w(c,\rad_w(c))$ implies $\ACE_w(\nu_w(d),d)$.
\qed\end{proof}

We will show that the function $\PCE$ satisfies the following precondition and postcondition,
where $w$ and $w'$ are the working strings before and after a call, respectively.
\begin{condition}[Precondition of $\PCE(c)$]\label{con:precon}{\ }
	\begin{itemize}
		\item $w$ is pp-irreducible,
		\item $\Stack$ is empty,
		\item $c + \rad_w(c) \ge |w|-1$,
		\item $c$ is left-good,
		\item For all positions $d \in [1: \mscr{A}_w(c)-1] \cup [c+1 : |w|-1]$, $d$ is stable in $w$ and $\Pals[d]=\rad_w(d)$.
	\end{itemize}
\end{condition}

\begin{condition}[Postcondition of $\PCE(c)$]\label{con:postcon}
 If it returns $\msf{true}$, then
 		\begin{itemize}
			\item $w' = \widehat{wu}$ for a string $u$ appended from the input such that $|w'| \le c$,
			\item $\Stack$ is empty.
		\end{itemize}
	If it returns $\msf{false}$, then
		\begin{itemize}
		 	\item $w'$ is a pp-irreducible string 
				 such that $\mscr{F}_{w'}(c) = |w'|-1$ and $wu \ZR^* w'$ for a string $u$ appended from the input,
			\item $(c;\Stack)$ is a palindrome chain such that $\mscr{F}_{w'}(c) = F_{w'}(c;\Stack)$,
			\item For all positions $d \in [1: \mscr{A}_{w'}(c)-1] \cup [c : |w'|-1]$, $d$ is stable in $w'$ and $\Pals[d]=\rad_{w'}(d)$.
	\end{itemize}
\end{condition}
\begin{lemma}[$\PCE$]\label{lem:stabilize}
Suppose that $c$ satisfies Condition~\ref{con:precon}.
Then after executing $\PCE(c)$, Condition~\ref{con:postcon} is satisfied.
\end{lemma}
A significant amount of the rest of this subsection is dedicated to proving Lemma~\ref{lem:stabilize}.
Assuming Lemma~\ref{lem:stabilize} being true, we establish the following proposition.
\begin{proposition}\label{prop:correct}
	Algorithm~\ref{alg:zred} calculates the Z-normal form of the input.
\end{proposition}
\begin{proof}
	Suppose that $\PCE(|w|-1)$ is called from the main function $\ZRed(T)$ and the working string becomes $w'$.
	We show by induction that positions in $[1:|w'|-1]$ are all strongly stable in $w'$ and $w'$ is the Z-normal form of the prefix of the input we have read so far.
	
	Suppose that positions in $[1:c-1]$ are all strongly stable when $\PCE(c)$ is called from $\ZRed(T)$, where $c = |w|-1$.
	This is vacuously true for the first call, where $c=1$.
	Suppose that $\PCE(c)$ returns $\msf{false}$.
	All positions in $[1:c-1]$ are strongly stable in $w$ and thus so are in $w'$ by Corollary~\ref{cor:stronglystable}
	and all positions in $[c:|w'|-1]$ are stable in $w'$ by the postcondition (Condition~\ref{con:postcon}).
	Therefore, positions in $[1:|w'|-1]$ are all strongly stable in $w'$ and $w'$ has no nonempty suffix palindrome, i.e., $w'$ is Z-normal.
	
	Suppose $\PCE(c)$ returns $\msf{true}$.
	This may happen on Lines~\ref{al:true1} and~\ref{al:true2}.
	Actually the former cannot happen, since this means that $c$ becomes a Z-shape's right pivot and the left pivot is less than $c$, which contradicts that all positions $d \in [1: c-1]$ are strongly stable in $w$ (Corollary~\ref{cor:stronglystable}).
	When the latter takes place, we have $\prev{w'}=\prev{w}$, so the induction hypothesis applies.
\qed\end{proof}

When $\msf{Stabilize}(c)$ tries to fix the value $\Pals[c]$ to be $\rad_w(c)$, the right arm of the palindrome at $c$ may be cut in the middle after finding the end of the right arm in a string. 
Then we need to extend it again. The \textbf{while} loop is repeated until $c$ becomes stable.
\begin{condition}[Precondition of the \textbf{while} loop]\label{con:whilecon}
	In addition to Condition~\ref{con:precon}, 
	\begin{itemize}
		\item for all positions $d \in [c+1 : |w|-1]$, $\Pals[d] = \Pals[2c-d]$.
	\end{itemize}
\end{condition}

In what follows we give some lemmas that explain the behavior of our algorithm in a more formal way.

\begin{lemma}[$\slowExt$]\label{lem:slow}
	Suppose that at the beginning of an iteration of the \textbf{while} loop of\/ $\PCE(c)$, Condition~\ref{con:whilecon} holds.
	Let $w$ and $w'$ be the working strings before and after execution of\/ $\slowExt(c)$, respectively.
	If\/ $\slowExt(c)$ returns $\msf{true}$, then
	\begin{itemize}
		\item $w' = \widehat{wu}$ for $u$ appended from the input such that $wu$ is ss-reducible and the right pivot of the Z-shape is $c$,
	\end{itemize}
	If\/ $\slowExt(c)$ returns $\msf{false}$,
	\begin{itemize}
		\item $w' = wu$ for $u$ appended from the input such that $c+\rad_{w'}(c) = |w'|-1$ and ${w'}$ is pp-irreducible,
		\item $\Pals[c] = \rad_{w'}(c)$,
		\item $\Pals[c+r] = \Pals[c-r]$ for all $r \in [1:\Pals[c]]$, 
		\item $c$ is left-good and right-good.
	\end{itemize}
\end{lemma}
\begin{proof}
By Lemma~\ref{lem:preslow} and the fact that $\prev{w}$ is irreducible, $\slowExt(c)$ returns $\msf{true}$ iff $w$ has a Z-shape with right pivot $c$, in which case $\slowExt(c)$ deletes the tail.
Suppose $\slowExt(c)$ returns $\msf{false}$.
By $c+\rad_w(c) \ge |w|-1$ and the behavior of $\slowExt(c)$, for all $q \in [c:|w'|-1]$, $c$ is the center of a suffix palindrome in $w'[1:q]$,
but no Z-shape had $c$ as its right pivot in $w'[1:q]$. Hence, $\prev{w'}$ has no Z-shape by Lemma~\ref{lem:uniqueZ}, i.e., $w'$ is pp-irreducible.

The algorithm lets $\Pals[c] = \rad_{w'}(c)$ and $\Pals[c+r] = \Pals[c-r]$ for all $r \in [|w|-c:\Pals[c]]$,
while  $\Pals[c+r] = \Pals[c-r]$ for $r \in [1:|w|-c-1]$ is guaranteed by Condition~\ref{con:whilecon}.
The right-goodness follows Lemmas~\ref{lem:leftgood} and~\ref{lem:left2right}.
Since $\nu_w(c) \plink_w c$ implies $\nu_w(c) \plink_{w'} c$, we have $\nu_{w'}(c) \ge \nu_w(c)$.
Hence, $c$ is left-good in $w'$.
\qed\end{proof}

\begin{lemma}[$\fastExt$]\label{lem:fast}
	Suppose that\/ $\fastExt(d)$ is called from\/ $\PCE(c)$ satisfying that
	\begin{itemize}
		\item $c$ is left-good and right-good,
		\item $\Pals[c]=\rad_w(c)$ and $\Pals[e]=\rad_w(e)$ for all $e \in [d+1 : |w|-1 ]$,
		\item either $\Stack$ is empty or $\Stack$ is a maximal palindrome chain from some $e > d$ such that $c \plink_w e$,
		\item $F_{w}(c;\Stack) = \max(\{\mscr{F}_w(e) \mid d < e \le c+\rad_w(c)\} \cup \{c+\rad_w(c)\}) = |w|-1$.
	\end{itemize}
	Then after the execution, 
	\begin{itemize}
		\item if it returns $\msf{true}$, then $d+\rad_{w}(d) \ge |w|-1$ and $\Stack$ is empty,
		\item if it returns $\msf{false}$, then
		\begin{itemize}
			\item $\Stack$ is not empty and $(d;\Stack)$ is a maximal palindrome chain such that $F_{w}(d;\Stack) = \mscr{F}_w(d) = |w|-1$, 
		 	\item $\Pals[e]=\rad_w(e)$ and $e$ is stable in $w$ for all $e \in [d: |w|-1]$.
		 \end{itemize}
	\end{itemize}
\end{lemma}
\begin{proof}
If $\Stack$ is empty when $\fastExt(d)$ is called, it immediately returns $\msf{true}$. Since $c$ is right-good, $d+\Pals[d] \ge c+\Pals[c]=c+\rad_w(c)$ implies $d+ \rad_w(d) \ge c+\rad_w(c)=|w|-1$.

Suppose that $\Stack = (c_1,\dots,c_k)$ for some $k \ge 1$. Let $r_i = c_i-d$ for $i \in [1:k]$.
By induction, we show that at the beginning of the $i$th iteration of the \textbf{while} loop with $\Stack = (c_i,\dots,c_k)$,
\begin{itemize}
	\item[(i)] $r_i \le \rad_{w}(d)$, 
	\item[(ii)] $r_i < d-c$, 
	\item[(iii)] $\rad_w(d-r_i) = \Pals[d-r_i]$,
\end{itemize}
unless $\fastExt(d)$ returns $\msf{false}$ earlier.

First we show that the above claims (i)--(iii) hold for $i=1$.
The assumption $c \plink_w c_1 = d + r_1$ implies that $d + r_1 \le c+\rad_w(c) \le d + r_1 +\rad_w(d + r_1)$.
Since $c$ is right-good, $d+ \Pals[d] \ge c+ \Pals[c]$ implies $d+ \rad_w(d) \ge c+ \rad_w(c)$.
Then we have $d + r_1 \le c+\rad_w(c) \le d+ \rad_w(d)$, i.e., $r_1 \le \rad_w(d)$.
This proves (i).
If (ii) did not hold, together with (i), we have $d-c \le r_1 \le \rad_{w}(d)$, in which case, $\lrangle{c,d}$ is a Z-shape in $w$.
Since $d+r_1 \le |w|-1$, this Z-shape occurs in $\prev{w}$, which contradicts that $w$ is pp-irreducible.
To show (iii) by contradiction, suppose $\Pals[d-r_1] \neq \rad_w(d-r_1)$.
By (ii), $c < d-r_1$.
Since $c$ is right-good, $\rad_w(d-r_1) \ge c+\rad_w(c)-(d-r_1) \ge r_1$.
This means $\lrangle{d-r_1,d}$ is a Z-shape in $w$. Since $d+r_1 \le |w|-1$, this Z-shape occurs in $\prev{w}$. Contradiction.

We assume the claims (i)--(iii) hold at the beginning of the $i$th iteration of the \textbf{while} loop.

Suppose that $\Pals[d-r_i] < \Pals[d+r_i]$, which means $\rad_w(d-r_i) < \rad_w(d+r_i)$ by (iii) and the assumption of the lemma.
In this case, the function returns $\msf{false}$ after letting $\Pals[d] = r_i+\rad_w(d-r_i)$.
By Lemma~\ref{lem:prefast}, $\rad_w(d) = r_i + \rad_w(d-r_i) < r_i + \rad_w(d+r_i)$,
which means $\Pals[d] = \rad_w(d)$.
Moreover, $\rad_w(d) < r_i + \rad_w(d+r_i)$ and (i) implies $d \plink_w d+r_i$.
This means that $\mscr{F}_w(d) \ge F_{w}(d;c_i,\dots,c_k)$.
It suffices to show $\mscr{F}_w(d) \le F_{w}(c_k)$.
Let $(d,d_1,\dots,d_m)$ be a maximal palindrome chain from $d$.
Note that $m \ge 1$, since $d+\rad_w(d) < c_i+\rad_w(c_i)$.
By $d_1 \le d+\rad_w(d) < |w|-1$, we have $d_1 \in [d+1:|w|-1]$.
Applying Corollary~\ref{cor:pchain2} to the assumption, we have $\mscr{F}_w(d) = \mscr{F}_w(d_1) \le \mscr{F}_w(c_1) = \mscr{F}_w(c_k) = |w|-1$.
%

Suppose that $\Pals[d-r_i] \ge \Pals[d+r_i]$, which means $\rad_w(d-r_i) \ge \rad_w(d+r_i)$ by (iii) and the assumption of the lemma.
Lemma~\ref{lem:prefast} implies 
\begin{equation}\label{eq:dri}
\rad_w(d) \ge r_i + \rad_w(d+r_i)
\,.
\end{equation}
If $i=k$, then the last element $c_k$ is popped from $\Stack$ and $\fastExt(d)$ returns $\msf{true}$.
By Eq.~\ref{eq:dri}, $d+\rad_w(d) \ge d+r_k + \rad_w(d+r_k) = F_w(c_k) = |w|-1$ holds.

For $i < k$, we must show that claims (i)--(iii) hold for $i+1$.
Recall that $c_i \plink_w c_{i+1}$ means $d + r_i + \rad_w(d+r_i) \ge d + r_{i+1}$,
 with which Eq.~\ref{eq:dri}, we have $ \rad_w(d) \ge  r_i + \rad_w(d+r_i) \ge r_{i+1}$. So (i) holds for $i+1$.
Then (ii) and (iii) follow (i) by the same argument for the case $i=1$ replacing $r_1$ with $r_i$.
%
\qed
\end{proof}

We have observed by Lemma~\ref{lem:slow} ($c$ is right-good) that when $\PCE(c)$ calls $\PCE(d)$, we have $c \plink_w d$,
and that for every position $e \in [c:|w|-1]$, there is $d$ such that $d \le e \le d+\rad_w(d)$ and there is a palindrome chain from $c$ to $d$.
This implies that if $\PCE(c)$ is called from $\ZRed(T)$, then $\mscr{A}_w(d)=c$ for every position $d \in [c:|w|-1]$ at any moment before $\PCE(c)$ terminates.
By Lemmas~\ref{lem:fast} and~\ref{lem:right2left}, when $\fastExt(d)$ returns $\msf{true}$, Condition~\ref{con:precon} for $d$ is satisfied, provided that the precondition of Lemma~\ref{lem:fast} is satisfied.

Now we have prepared enough for analyzing the function $\PCE(c)$.
Our goal is to show that Condition~\ref{con:postcon} holds for $\PCE(c)$ provided that Condition~\ref{con:precon} holds.
The function $\PCE(c)$ calls $\PCE(d)$ recursively.
For now we assume that Condition~\ref{con:precon} implies Condition~\ref{con:postcon} for those $d$.
Then this inductive argument completes a proof of Lemma~\ref{lem:stabilize}.

Suppose that Condition~\ref{con:precon} holds for $\PCE(c)$.
If $\slowExt(c)$ returns $\msf{true}$, clearly Condition~\ref{con:postcon} holds by Lemma~\ref{lem:slow}.
Hereafter we suppose that $\slowExt(c)$ returns $\msf{false}$.

\begin{lemma}[\textbf{for} loop]\label{lem:for}
	Suppose that Condition~\ref{con:whilecon} is satisfied at the beginning of every iteration of the \textbf{while} loop.
	Then, at the beginning of each iteration of the \textbf{for} loop of\/ $\PCE(c)$, the following holds.
	\begin{itemize}
		\item[(i)] $(c;\Stack)$ is a palindrome chain such that  
		\[
		F_w(c;\Stack) = \max(\{\,\mscr{F}_w(e) \mid d < e \le c+\rad_w(c)\,\} \cup \{c+\rad_w(c)\}) = |w|-1,
		\]
		\item[(ii)] $c$ is left-good and right-good,
		\item[(iii)] for all $e \in [c+1: d]$, $\Pals[d] = \Pals[2c-d]$,
		\item[(iv)] $\Pals[c]=\rad_w(c)$ and $\Pals[e]=\rad_w(e)$ for all $e \in [d+1: |w|-1]$.
	\end{itemize}
	Moreover if we \textbf{break} the loop on Line~\ref{al:break}, still Condition~\ref{con:whilecon} holds.
	If we return $\msf{true}$ on Line~\ref{al:true2}, Condition~\ref{con:postcon} holds for $c$.
\end{lemma}
\begin{proof}
	For $d = c+\Pals[c]$, the lemma follows Lemmas~\ref{lem:slow} and Condition~\ref{con:whilecon}.
	We show the lemma holds for $d-1$ if it is the case for $d$.
	
	If $d + \Pals[d] < c + \Pals[c] = c+ \rad_w(c)$, the algorithm does nothing but decreasing the value of $d$.
	No need to prove (ii) and (iii).
	Since $c$ is right-good, we already have $\Pals[d]=\rad_w(d)$.
	With the induction hypothesis this shows (iv).
	Since $d+\rad_w(d) < c + \rad_w(c) \le |w|-1$, we have
	\begin{align*}
	\mscr{F}_w(d) &= \max(\{\mscr{F}_w(e) \mid d \plink_w e\} \cup \{ F_w(d) \})
	\\ & \le \max(\{\mscr{F}_w(e) \mid d < e \le c+ \rad_w(c) \} \cup \{ c + \rad_w(c) \}) = |w| - 1 
	\,.\end{align*}
	This proves (i).
	
	Suppose $d + \Pals[d] \ge c + \Pals[c]$.
	If $\fastExt(d)$ returns $\msf{false}$, Lemma~\ref{lem:fast} implies that $\Pals[e]=\rad_w(e)$ for all $e \in [d : |w|-1]$, which means (iv),
	and that $(d;\Stack)$ is a palindrome chain such that $F_{w}(d;\Stack) = \max\{\,\mscr{F}_w(e) \mid d \le e \le c+\Pals[c]\,\} = |w|-1$ and $\Stack$ is not empty.
	We then push $d$ to the stack.  For clarity, we write the updated stack as $\Stack' = (d;\Stack)$ here.
	By $c \plink_w d$, $(c;\Stack')$ is a palindrome chain such that $F_w(c;\Stack') = F_w(\Stack)  = |w|-1$.
	This proves (i).
	Since the procedure changes $\Pals$ only at $d$ as $\Pals[d]=\rad_w(d)$, (ii) and (iii) are obvious.
	
	Now suppose that $\fastExt(d)$ returns $\msf{true}$.
	Then $\PCE(d)$ will be called.
	We first confirm that Condition~\ref{con:precon} for $d$ is satisfied.
	Since $\fastExt(d)$ returns $\msf{true}$, Lemma~\ref{lem:fast} implies that $\Stack$ is empty and $d+\rad_w(d) \ge |w|-1$.
	Together with $|w| - 1 \ge c+\rad_w(c)$, we have $c \plink_w d$.
	Since $c$ is right-good, $d$ is left-good by Lemma~\ref{lem:right2left}.
	By induction hypothesis, $\Pals[e]=\rad_w(e)$ for all $e \in [d+1:|w|-1]$.
	Recall that $c \plink_w d$ implies $\mscr{A}_w(c) = \mscr{A}_w(d)$. 
	So, all positions in $[1:\mscr{A}_w(d)-1]$ are stable.
	For positions $e \in [d+1:|w|-1]$, the assumption (i) implies $\mscr{F}_w(e) \le |w|-1$ by Corollary~\ref{cor:pchain2}, so they are stable.
	Therefore, Condition~\ref{con:precon} for $d$ is satisfied 
	 and thus we may assume that Condition~\ref{con:postcon} for $d$ is satisfied.
	
	Suppose $\PCE(d)$ returns $\msf{false}$.
	Then, $\Pals[e] = \rad_{w'}(e)$ for all $e \in [d : |w'|-1]$,
	where $w'$ is the working string after the execution of $\PCE(d)$, i.e., (iv) holds.
	Since $\PCE(d)$ does not change the value of $\Pals[e]$ for $e < d$, (ii) and (iii) hold.
	Moreover, $(d;\Stack)$ is a palindrome chain such that $F_{w'}(d;\Stack) = |w'|-1$.
	By pushing $d$ to the stack, (i) holds.
	
	If $\PCE(d)$ returns $\msf{true}$,  we exit the \textbf{for} loop.
	By Condition~\ref{con:postcon} on $d$, the stack is empty and $|w_1| \le d$ where $w_1$ is the working string after the execution of $\PCE(d)$.
	$\PCE(d)$ returns $\msf{true}$ if and only if either $\slowExt(d)$ returns $\msf{true}$ (Line~\ref{al:true1}) or $|w_1| = d$ (Line~\ref{al:true2}).
	
	Suppose $\slowExt(d)$ returns $\msf{true}$ by detecting and contracting a suffix Z-shape occurrence $\lrangle{2d - |w_0|,d}$ in some $w_0$ such that $\widehat{w_0} = {w_1} = w_0[1:2d-|w_0|]$.
	Here one can see $2d-|w_0| \ge c$, since otherwise, $\lrangle{c,d}$ was an occurrence of another Z-shape in $w_0[1 : 2d-c]$, where $2d-c < |w_0|$, which $\slowExt(d)$ should have detected and contracted earlier.
	Thus, $|w_1| \ge c$.
	If $|w_1|=c$, $\PCE(c)$ returns $\msf{true}$, too, on Line~\ref{al:true2}.
	Condition~\ref{con:postcon} on $d$ immediately implies Condition~\ref{con:postcon} for $c$, with the fact $|w_1|=c$.
	Suppose $c < |w_1| < d$.
	In this case, we break the \textbf{for} loop and iterate the \textbf{while} loop after
	 appending a new letter from $T$ to $w_1$ on Line~\ref{al:append}.
	Since $d \le c+\rad_w(c)$, we have $c+\rad_{w_1}(c) = |w_1|$.
	Therefore, since $\PCE(d)$ does not change the values of $\Pals[e]$ for $e < d$ unless $w[e]$ is deleted,
	after appending a letter to $w_1$, Condition~\ref{con:precon} shall be satisfied.
	
	Suppose $|w_1| = d$.
	In this case, we let $\Pals[d] = \Pals[2c-d]$ before breaking the \textbf{for} loop.
	$\PCE(d)$ may have changed the value of $\Pals[d]$ by calling $\slowExt(d)$.
	The value is restored to be $\Pals[d] = \Pals[2c-d]$ so that Condition~\ref{con:whilecon} for $c$ still holds.
	Apart from this point, every requirement of Condition~\ref{con:whilecon} on the next iteration follows from that on the current iteration.
\qed\end{proof}

\begin{lemma}[\textbf{while} loop]\label{lem:while}
	Suppose that Condition~\ref{con:precon} is satisfied when $\PCE(c)$ is called.
	Then at the beginning of every iteration of the \textbf{while} loop, Condition~\ref{con:whilecon} holds.
	Moreover if it returns $\msf{true}$, Condition~\ref{con:postcon} holds.
\end{lemma}
\begin{proof}
	At the first iteration of the loop, 
	since $c+\rad_w(c) \ge |w|-1$ and $\Pals[c+r]=\rad_w(c+r)$ for each $r \in [1:|w|-c-1]$,
	it is enough to show that $\rad_w(c-r)=\Pals[c-r]$ for each $r \in [1:|w|-c-1]$.
	This can be seen by Lemma~\ref{lem:leftgood}.
	For the second or later iteration, we must have broken the \textbf{for} loop in the previous iteration.
	Condition~\ref{con:whilecon} follows Lemma~\ref{lem:for}.
	
	If it returns $\msf{true}$ on Line~\ref{al:true1}, $\Stack$ is empty due to Condition~\ref{con:whilecon}.
	 Lemma~\ref{lem:slow} ensures the other requirement of Condition~\ref{con:postcon}.
	If it returns $\msf{true}$ from the \textbf{for} loop on Line~\ref{al:true2}, Lemma~\ref{lem:for} ensures Condition~\ref{con:postcon}.
\qed\end{proof}

\begin{theorem}\label{thm:main}
	Our algorithm calculates the Z-normal form of the input in linear time.
\end{theorem}
\begin{proof}
	We first show that the number of calls of the function $\PCE$ is bounded by $|T|$.
	For this sake, we associate each occurrence of a letter in $w$ with the original position in $T$ and
	let $\tilde{w} \in (\Sigma \times \mathbb{N})^*$ denote the string obtained from $w$ by adding the original position to each letter.
	We assume that $(T[i],i)$ is \emph{fresh} at the beginning
	and it becomes \emph{non-fresh} when $\PCE(d)$ is called and $\tilde{w}[d+1]=(T[i],i)$.
	We claim that $\PCE(d)$ is called only when $\tilde{w}[d+1]$ is fresh, which implies that the number of calls of the function $\PCE$ is bounded by $|T|$.

	When $\PCE(d)$ is called from $\ZRed(T)$, it is right after a new letter is appended at position $d + 1$, which must be fresh.
	Suppose $\PCE(d)$ is called from $\PCE(c)$.
	It suffices to show that all letters in $\tilde{w}[b+1:|w|]$ are fresh before entering the \textbf{for} loop of $\PCE(c)$.
	This claim is obviously true if it is at the first iteration%
		\footnote{Precisely speaking, ``the first iteration'' means the first execution of instructions in the \textbf{while} loop, before ``iterating'' the loop.} of the \textbf{while} loop,
	 since at the beginning of the $\PCE(c)$, $[b+1:|w|]=\emptyset$ and after the execution of $\slowExt(c)$, all letters in $\tilde{w}[b+1:|w|]$ have just been newly appended and are fresh.
	Suppose that the claim holds at the beginning of an iteration of the \textbf{while} loop.
	The claim still holds after the execution of $\slowExt(c)$ by the same reason for the first iteration.
	The \textbf{while} loop will be repeated only when $\PCE(e)$ returns $\mathsf{true}$ for some $e \in [b:c+\Pals[c]]$ in the \textbf{for} loop in $\PCE(c)$, in which case the working string $w$ is reduced and its length becomes $e$ or smaller by Condition~\ref{con:postcon}.
	That is, all non-fresh letters on the right of $e$ are deleted and letters that will be appended are all fresh.
	This completes proving that the number of calls of the function $\PCE$ is bounded by $|T|$.
	
	The above explanation about the number of calls of $\PCE$ also shows that the total number of iterations of the \textbf{while} loop is bounded by $|T|$
	and this implies the number of calls of $\slowExt$ is also bounded by $|T|$.
	The total running time of $\slowExt$ is linearly bounded by the number of its calls plus the times of appending letters from $T$, which is bounded by $\mrm{O}(|T|)$ in total.
	The same argument on the number of calls of $\PCE$ applies to that of executions of the \textbf{for} loop.
	This implies that the total number of positions that is pushed onto the stack is bounded by $|T|$, which implies that total running time of $\fastExt$ is bounded by $\mrm{O}(|T|)$.
	
	All in all, our Z-reduction algorithm runs in linear time.
\qed\end{proof}

By using our and Raghavan's~\cite{RAGHAVAN1994108} algorithms,
the smallest path and cycle can be inferred from walks in linear time.

\begin{corollary}
	Given a string $T$ of length $n$,
	the smallest path and cycle on which $T$ is the output of a walk
	can be inferred in $\mrm{O}(n)$ time.
\end{corollary}

\section{Experiments}
This section presents experimental performance of our algorithm
 comparing with Raghavan's $\mrm{O}(n \log n)$ time algorithm~\cite{RAGHAVAN1994108}.
We implemented these algorithms in C++ and compiled with Visual C++ 12.0 (2013) compiler.
The experiments were conducted on Windows 7 PC with Xeon W3565 and 12GB RAM.
In the whole experiments, we got the average running time for $10$ times of attempts.

First, for randomly generated strings of length between $10^5$ and $10^6$ over $\Sigma$ of size $|\Sigma|=2,6,10$,
we compared the running time of the algorithms (Fig.~\ref{fig:exp_la}~(a)).
For any alphabet size, our proposed algorithm ran faster.
\begin{figure}[tb]
	\begin{subfigure}[l]{0.48\textwidth}
		\includegraphics[scale=0.45]{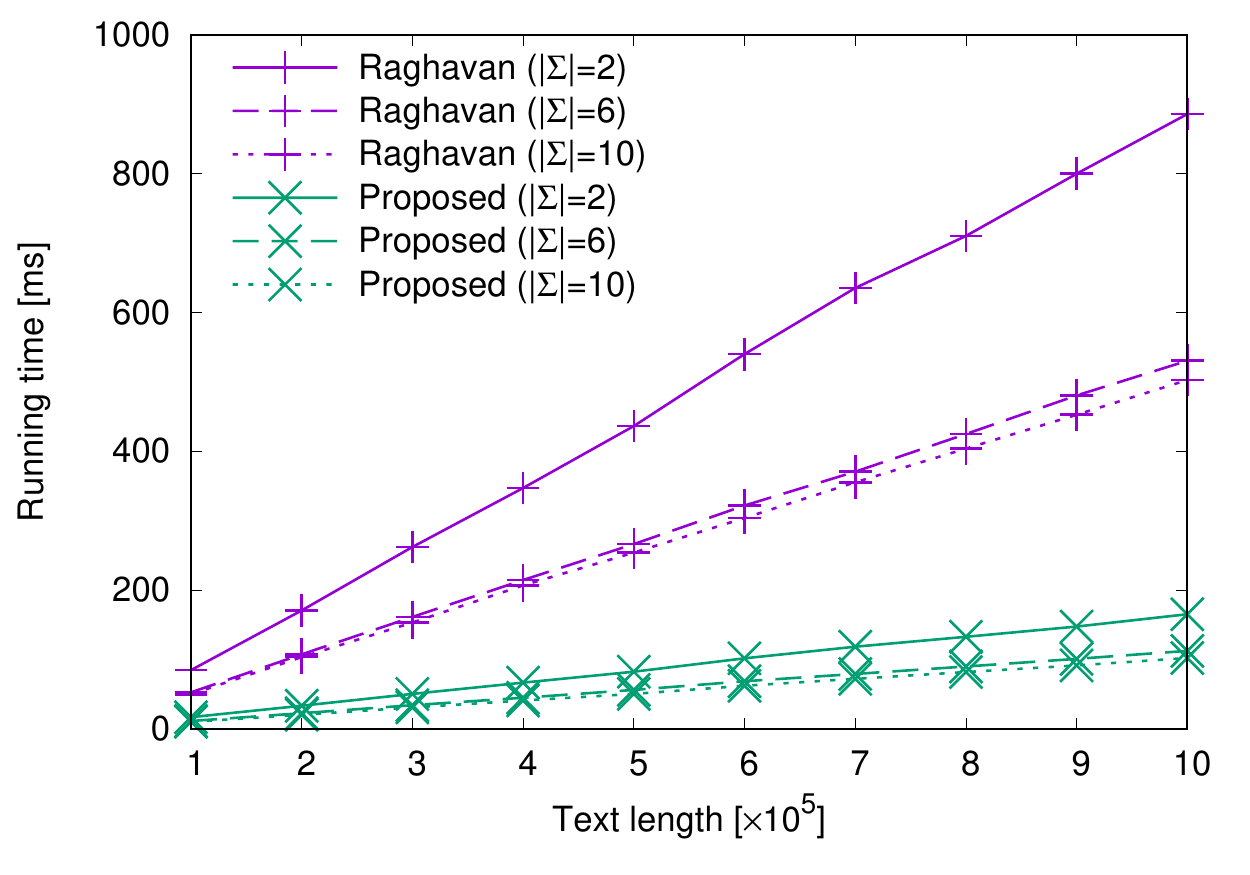}
	\caption{For length between $10^5$ and $10^6$}
	\end{subfigure}
	\begin{subfigure}[l]{0.48\textwidth}
		\includegraphics[scale=0.45]{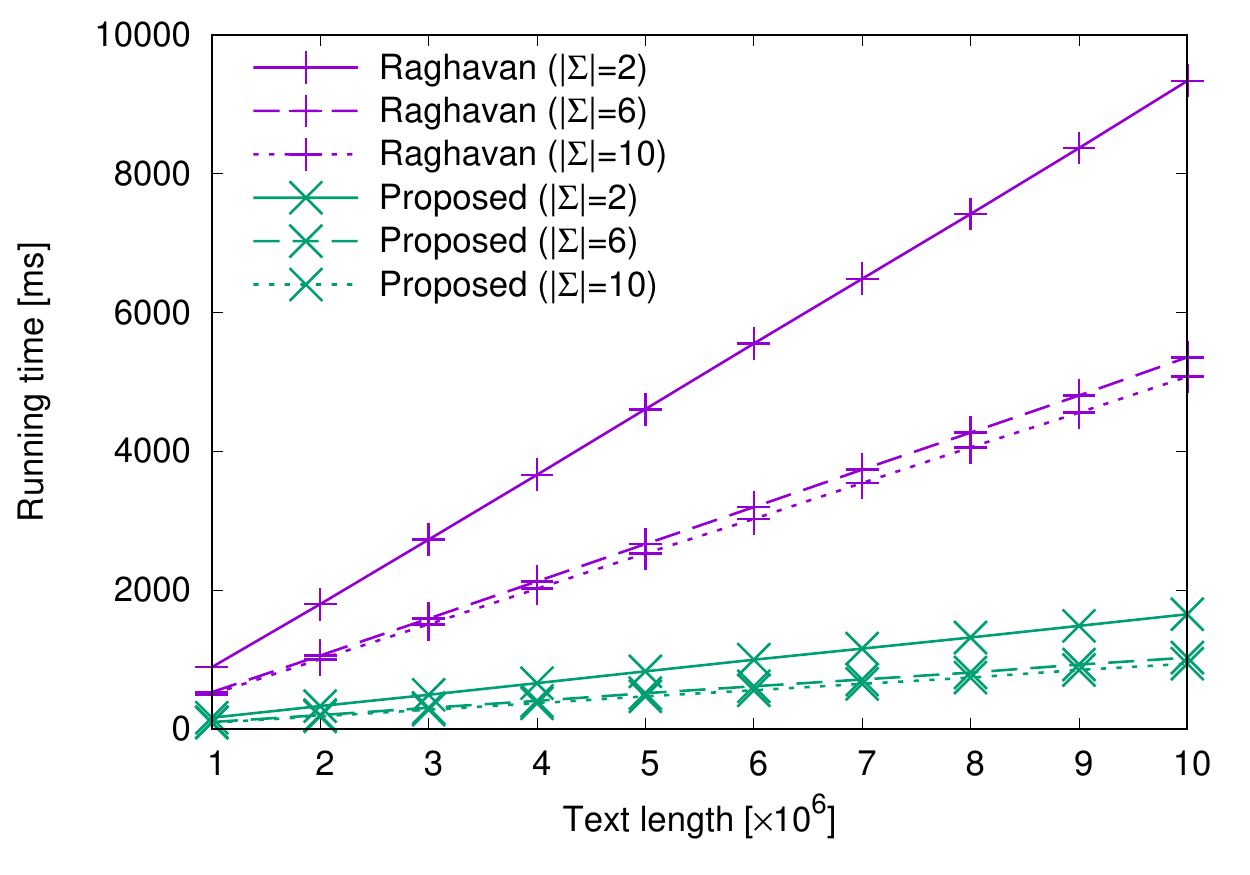}
	\caption{For length between $10^6$ and $10^7$}
	\end{subfigure}
	\caption{Running time of our Z-reduction algorithm for random strings with $|\Sigma|=2,6,10$}
	\label{fig:exp_la}
\end{figure}

Furthermore, we conducted experiments for strings of length between $10^6$ and $10^7$ with the same alphabets
and got a similar result (Fig~\ref{fig:exp_la}~(b)).
Here, the slope of Raghavan's algorithm's performance increases slightly as the string length increases.
On the other hand, our proposed algorithm keeps the same slope.
This shows the proposed algorithm runs in linear time in practice.  


\section*{Acknowledgments}
The research is supported by JSPS KAKENHI Grant Numbers JP15H05706, JP26330013 and JP18K11150, and ImPACT Program of Council for Science, Technology and Innovation (Cabinet Office, Government of Japan).

\bibliographystyle{elsarticle-num} 
\bibliography{ref}

\begin{thebibliography}{1}
\expandafter\ifx\csname url\endcsname\relax
  \def\url#1{\texttt{#1}}\fi
\expandafter\ifx\csname urlprefix\endcsname\relax\def\urlprefix{URL }\fi
\expandafter\ifx\csname href\endcsname\relax
  \def\href#1#2{#2} \def\path#1{#1}\fi

\bibitem{Aslam1990359}
J.~A. Aslam, R.~L. Rivest, Inferring graphs from walks, in: Computational
  Learning Theory, 1990, pp. 359--370 (1990).

\bibitem{RAGHAVAN1994108}
V.~Raghavan, Bounded degree graph inference from walks, Journal of Computer and
  System Sciences 49~(1) (1994) 108--132 (1994).

\bibitem{MARUYAMA1995257}
O.~Maruyama, S.~Miyano, Graph inference from a walk for trees of bounded degree
  3 is {NP}-complete, in: Mathematical Foundations of Computer Science 1995,
  1995, pp. 257--266 (1995).

\bibitem{MARUYAMA1996289}
O.~Maruyama, S.~Miyano, Inferring a tree from walks, Theoretical Computer
  Science 161~(1) (1996) 289--300 (1996).

\bibitem{Akutsu2005371}
T.~Akutsu, D.~Fukagawa, Inferring a graph from path frequency, in:
  Combinatorial Pattern Matching, 2005, pp. 371--382 (2005).

\bibitem{Manacher75}
G.~K. Manacher, A new linear-time on-line algorithm for finding the smallest
  initial palindrome of a string, J. {ACM} 22~(3) (1975) 346--351 (1975).

\bibitem{Narisada2018}
S.~Narisada, D.~Hendrian, R.~Yoshinaka, A.~Shinohara, Linear-time online
  algorithm inferring the shortest path from a walk, in: String Processing and
  Information Retrieval, Springer International Publishing, 2018, pp. 311--324
  (2018).

\end{thebibliography}
\end{document}